\documentclass{article}
\usepackage{booktabs} 
\usepackage[ruled]{algorithm2e} 

\SetAlFnt{\small}
\SetAlCapFnt{\small}
\SetAlCapNameFnt{\small}
\SetAlCapHSkip{0pt}
\IncMargin{-\parindent}

\usepackage{fullpage}

\usepackage{natbib}
\setcitestyle{authoryear}

\SetCommentSty{mycommentfont}
\SetEndCharOfAlgoLine{}

\usepackage{amsfonts}  
\usepackage{amsthm,amsmath, amssymb, xparse}
\usepackage{xcolor}
\usepackage{tikz}
\usepackage{pgfplots}
\pgfplotsset{compat=1.17} 
\usepackage{graphicx}
\usepackage{caption}
\usepackage{subcaption}
\usepackage{float}
\usepackage{url}
\usepackage[breaklinks]{hyperref}
\newif\ifshowcomments
\showcommentsfalse   
\newcommand{\docomment}[3]{\ifshowcomments \textcolor{#1}{[ #2 : #3 ]} \fi}
\newcommand{\bo}[1]{\docomment{brown}{Bo}{#1}}
\newcommand{\maneesha}[1]{\docomment{blue}{Maneesha}{#1}}

\newtheorem{theorem}{Theorem}
\newtheorem{prop}{Proposition}
\newtheorem*{prop*}{Proposition}
\newtheorem{lemma}{Lemma}

\newtheorem{corollary}{Corollary}
\newtheorem*{corollary*}{Corollary}
\newtheorem{example}{Example}
\newtheorem{definition}{Definition}

\theoremstyle{definition}\newtheorem{remark}{Remark}
\theoremstyle{definition}

\DeclareMathOperator*{\E}{\mathbb{E}}
\newcommand{\reals}{\mathbb{R}}
\renewcommand{\P}{\mathcal{P}}

\renewcommand{\bar}[1]{\overline{#1}}

\usepackage{amsmath,xparse,zref-savepos}

\makeatletter
\@ifundefined{zsaveposx}{\let\zsaveposx\zsavepos}{}
\newcounter{hposcnt}
\renewcommand*{\thehposcnt}{hpos\number\value{hposcnt}}
\NewDocumentCommand{\lplabel}{o m}{%
  \stepcounter{hposcnt}%
  \zsaveposx{\thehposcnt l}%
  \zref@refused{\thehposcnt l}%
  \zref@refused{hpos0l}%
  \makebox[0pt][r]{\makebox[\dimexpr\zposx{\thehposcnt l}sp-\zposx{hpos0l}sp][l]{#2}}%
  \IfNoValueF{#1}
    {\def\@currentlabel{#2}\ltx@label{#1}}
}
\makeatother
\AtBeginDocument{\zsaveposx{hpos0l}}
\begin{document}

\title{Contracts with Information Acquisition, via Scoring Rules}
\author{Maneesha Papireddygari \footnote{University of Colorado, Boulder, \texttt{mapa2091@colorado.edu}} \and Bo Waggoner \footnote{University of Colorado, Boulder, \texttt{bwag@colorado.edu}}}  

\maketitle

\begin{abstract}
  We consider a principal-agent problem where the agent may privately choose to acquire relevant information prior to taking a hidden action. 
  This model generalizes two special cases: a classic moral hazard setting, and a more recently studied problem of incentivizing information acquisition (IA).
  We show that all of these problems can be reduced to the design of a proper scoring rule.
  Under a limited liability condition, we consider the special cases separately and then the general problem.
  We give novel results for the special case of IA, giving a closed form ``pointed polyhedral cone'' solution for the general multidimensional problem.
  We also describe a geometric, scoring-rules based solution to the case of the classic contracts problem.
  Finally, we give an efficient algorithm for the general problem of Contracts with Information Acquisition.
\end{abstract}

\section{Introduction}
In principal-agent problems, a principal (e.g. employer) delegates a task to an agent (e.g. employee).
The agent in such cases generally has access to some information that the principal does not.
In one setting, the agent chooses an action that influences an outcome of interest.
The principal cannot directly observe the agent's actions and can only pay based on the observable outcome.
This is termed a \emph{hidden action} or \emph{moral hazard} model, and the design of contracts in this setting has been studied extensively in economics (e.g. \citet{mas1995microeconomic}) and, recently, algorithmic game theory (e.g. \citet{dutting2019simple}).
In a second setting, the agent can gather relevant information not available to the principal about the outcome of interest.
The principal again cannot directly observe this information, but can only judge its approximate quality based on the observable outcome.
Recent work has investigated the design of \emph{proper scoring rules}~\cite{savage1971elicitation,gneiting2007strictly} to incent the agent to \emph{acquire} information and report it~\citep{li2020optimization,chen2021optimal}.

\vskip1em
\noindent This paper considers a principal-agent problem featuring both types of asymmetry:
\begin{enumerate}
  \item A principal delegates a task to the agent.
  \item The agent can first choose to acquire costly information that is relevant to the final outcome and the best course of action.
  \item The agent agrees to a contract with the principal.
  \item Then, the agent selects and carries out a costly action.
  \item The principal does not observe the action or whether information acquisition takes place, but only observes a final outcome, a noisy function of both information and action.
        The principal pays the agent according to the contract, which is based on the observed outcome alone.
\end{enumerate}

To incent the agent, the principal offers in Step 1 a \emph{menu} of contracts, each a function mapping the observed outcome to the agent's payment.
The agents selects a contract from the menu in Step 3.
The principal designs the menu in order to incentivize this hidden information acquisition as well as to incentivize the desired choice of hidden action conditioned on the information.

\paragraph{Examples.}
We modify two examples, given by \citet{dutting2019simple} for the pure hidden-action setting, to this setting that also involves information acquisition.
First consider a clothing company (principal) employing a marketing firm (agent).
The marketing firm can choose to conduct a costly survey to learn customer preferences.
Then, regardless of whether it acquired the information or not, it will run a marketing campaign, in which it has several choices of effort level and strategy, which may be influenced by the survey results.
The clothing company wishes to design contracts based only on objective observables, such as quantity of sales, to incent the firm to conduct the survey and select an appropriate strategy based on the survey results.

Second, consider an insurance company (principal) designing a contract with a vehicle owner (agent).
The owner can take a (costly) educational course, e.g. on safe driving or on maintenance.
Then, the owner makes repair and driving style decisions, possibly based on the information acquired in the course.
These influence the possible financial outcomes, which are borne by the insurance company.
The insurance company wishes to design contracts that incent the agent to learn and then practice good maintenance and safe driving.

\paragraph{Limited liability.}
We will impose a standard \emph{limited liability} constraint on the principal's choice of contracts: the agent must never be required to make a payment to the principal.
Limited liability, often used in the contracts setting (e.g. \citet{bolton2004contract,tadelis2005lectures}), plays a role similar to risk aversion of the agent.
Without such an assumption, it is often possible for the principal to pay zero in expectation and transfer all risk onto the agent. However, these restrictions are often more realistic and give rise to a more interesting problem \citep{grossman1983analysis}.


\subsection{Our results}
\paragraph{Reducing the design space: proper scoring rules.}
Because the design space for contract menus is large, our first step is to reduce it.
We observe (Section \ref{sec:scoring-rules}) that the general problem, perhaps surprisingly, reduces without loss of generality to the design of a \emph{proper scoring rule}: a function $s(p,\omega)$ that assigns a payment for a prediction $p$ when the true outcome turns out to be $\omega$.
While this observation has been made in the information acquisition literature at least as far back as \cite{babaioff2011only}, we have not found it in conjunction with hidden actions.
Scoring rules are discussed in Section \ref{sec:scoring-rules}.
The main intuition is that selecting a contract can be viewed as the agent revealing their belief or prediction (by revealed preference); by a revelation principle, any contract menu can therefore be simulated with a truthful one, i.e. a proper scoring rule.


We then separately apply the scoring rule perspective to the two special cases.
\paragraph{Information acquisition.}
If there are no hidden actions, i.e. we remove Step 4, then our setting reduces to Information Acquisition (IA).
Here, the agent's choice of contract simply reveals the outcome of the signal they acquired, so this models a principal paying an agent to collect and forecast based on hidden information.
We show (Section \ref{sec:info_acq}) that in our setting, an optimal scoring rule takes the form of a \emph{polyhedral pointed cone}.
We furthermore give a closed-form solution, a maximum over a number of carefully chosen ``indicator'' contracts.

\paragraph{Hidden action.}
The second special case, the classic hidden-action problem or contracts with moral hazard (Section \ref{sec:contract-theory}), is well known.
However, we use the scoring-rule perspective to provide some geometric characterizations of feasible and optimal solutions that may be of interest.
In particular, while it is classical that some menu consisting of a single contract is always optimal, we show that this optimal contract must be a \emph{shifted subtangent} of the \emph{convexified cost curve}, and we geometrically characterize the set of contracts that can be added to the menu without compromising optimality.
These results shed some light on necessary conditions for the structure of the optimal solution to our general problem.

\paragraph{The general problem.}
Finally, we return to the general problem (Section \ref{sec:general}).
While it appears impossible to obtain a closed-form solution -- even in the contracts setting, the solution involves a linear program (LP) -- we give an efficient algorithm for solving it via an LP.
The key point of the proof is that proper scoring rules are captured by convex functions, and we are able to show that a piecewise linear convex function with a small number of pieces is optimal.
Utilizing insights from the special cases, we reduce the dimension of the LP considerably.

\subsection{Related work} \label{sec:related-work}
\paragraph{Information acquisition.}
Prior work of \citet{li2020optimization,chen2021optimal} considered similar problems, namely, designing proper scoring rules to maximize the incentive of an agent to acquire a signal before reporting a prediction.
Our problem is framed slightly differently, as the cost of acquiring the signal is given and the goal is minimizing expected payment.
More significant is that we have different objective.
\citet{chen2021optimal} maximizes worst-case payoff difference between acquiring signal and not.
\citet{li2020optimization}) maximizes the difference between expected score for acquiring and reporting the signal posterior and reporting the prior, while having strict upper and lower bounds on scores.
All of these have somewhat different use-cases.
Ours is aimed at a setting with strict limited liability constraints, but no upper bound on the maximum payment.
We mention that these works each provide a number of additional results and investigations, but omit further discussion because those are less related to this work.


Both papers generally find that inverted cones are optimal, which will also be the case in our setting.
Some differences are that our result is more general than \citet{li2020optimization}, which provides a solution for multidimensional $\Omega$ only with a symmetric prior; and that the solution to our problem is more complex than in \citet{chen2021optimal}, where it is an inverted pyramid formed by a point at $(p_0,0)$ and corners at $(\delta_{\omega},1)$ for each $\omega$.
We visualize a comparison of the solutions obtained under the different assumptions in Figure \ref{fig:compare-related-work}.

\paragraph{Contract theory.}
There is of course a long line of work in economics on contract theory; we refer to e.g. \citet{mas1995microeconomic,grossman1983analysis} as a starting point.
Recent work has also investigated contracts from a robustness standpoint~ \citep{carroll2015robustness,dutting2019simple}; where agents have unknown types or parameters~\citep{alon2021contracts, guruganesh2021contracts}; and some complexity results~\citep{duetting2020complexity,azar2018computational}

While the focus of the work is not the original contract model, we do give an apparently-new perspective of contracts as scoring rules.
It is known, e.g. \citet{dutting2019simple}, that the minimum payment problem with limited liability can be solved by linear programming, but according to that source, ``Fairly little is known about the \emph{structure} of the optimal contracts that come out of this approach.''
We use the scoring rule perspective to geometrically characterize the feasible and optimal solutions to the contracts problem with LL, based on the notion of the \emph{convexified cost curve}.

\paragraph{Scoring rules and decisionmaking.}
\citet{oesterheld2020minimum} considers incentivizing one or more agents to acquire costly information and report it to a principal, who wishes to make a decision.
That paper takes a regret-based approach and considers robust contract design, e.g. linear contracts.
\citet{oesterheld2020decision} utilizes proper scoring rules to solve a different kind of principal-agent problem.
There, the agent provides the principal a prediction and a recommended plan, which the principal implements.
Unlike this paper (where the principal knows the entire prior distribution), for \citet{oesterheld2020decision} the principal knows nothing and must provide good incentives in the worst case over all priors.
Other papers that involve scoring rules and decisionmaking include \citet{boutilier2012eliciting}, where an expert who reports a prediction tries to influence the principal's resulting decision; and \citet{bacon2012predicting}, in which an agent predicts the time their task will be completed and also decides how much effort to exert on the task.

\section{ Preliminaries }\label{sec:prelim}

We define the \emph{Contracts with Information Acquisition} setting.
We first define the relevant variables, then give the order of events, then describe the information structure.
Then we formalize the \emph{minimum payment problem}.
Intuitively, the agent will be maximizing expected utility, and the problem will be, given a \emph{plan} for the agent, to design a contract menu that incentivizes the agent to follow that plan as cheaply as possible.

\paragraph{Variable definitions.}
Recall that the setting involves hidden information (the ``signal''), a hidden action, and an observable outcome (the ``event'').
The signal is a random variable $S$ taking values in $\Sigma$.
We often use $\sigma$ to denote realizations of $S$.
There is a cost $\kappa \geq 0$ for the agent to acquire the signal.
The actions available to the agent form a set of exhaustive and mutually exclusive actions $A$ with a generic action denoted $a \in A$.
Each action $a$ has an associated cost $c_a \geq 0$.
The observable event is a random variable $W$ taking values in a set $\Omega$ with a generic outcome denoted $\omega \in \Omega$.
We assume that $\Omega$, $\Sigma$, and $A$ are finite sets.

A \emph{contract} is a function $t: \Omega \to \reals$ that represents a commitment for the principal to pay the agent $t(W)$ after observing $W$.
Play proceeds as follows:
\begin{enumerate}
    \item The principal offers a menu (i.e. set) of contracts $T \subseteq \reals^{\Omega}$.
    \item The agent privately decides whether or not to acquire the signal $S$.
          If acquiring, they incur a cost of $\kappa$ and observe the realization of $S$; otherwise, nothing happens.
    \item The agent selects and signs a contract $t^* \in T$.
          The principal observes $t^*$.
    \item The agent privately selects an action $a^* \in A$, incurring a cost $c_{a^*}$.
    \item Both players observe the event $W$ and the principal pays the agent $t^*(W)$.
\end{enumerate}

\paragraph{Information structure.}
The set of probability distributions over a set $X$ is denoted $\Delta_X$.
The signal $S$ is distributed according to a common-knowledge prior $q \in \Delta_{\Sigma}$.
Nature draws $S$ according to $q$ at the beginning of Step 2 (regardless of whether the agent observes it).

The event $W$ is drawn from a distribution that is determined by the action $a^*$ taken by the agent as well as the signal $S$.
For each $a \in A$, and $\sigma \in \Sigma$, when $S=\sigma$ and $a^* = a$, the event is distributed according to $p_{a,\sigma} \in \Delta_{\Omega}$.
Nature draws a realization of $W$ from $p_{a^*,S}$ at the beginning of Step 5.

Under these assumptions, an agent who takes an action $a^*=a$ but did not observe the realization of $S$ has a belief $p_a := \E_{S} [p_{a,S}] = \sum_{\sigma \in \Sigma} q(\sigma) p_{a,\sigma}$.
(We assume the agent is Bayesian and rational and knows all parameters of the problem.)
Given the principal's menu $T$, the agent's objective is to maximize expected utility, which is expected payment minus expected costs incurred.
If the agent acquires the signal $S$ and chooses $t^*,a^*$ based on $S$, then her expected utility is
  \[ \text{agent utility} = \E_{S \sim q} \left[ \E_{W \sim p_{a^*,S}}[t^*(W)] - c_{a^*} - \kappa \right] . \]
If the agent does not acquire the signal and chooses $t^*,a^*$ independently of $S$, then her expected utility is
  \[ \text{agent utility} = \E_{S \sim q} \left[ \E_{W \sim p_{a^*,S}}[t^*(W)] \right] - c_{a^*} . \]

\paragraph{The minimum payment problem.}
Our goal is to incent the agent to follow a specific ``plan'' as cheaply as possible.
Formally, a \emph{plan} consists of a decision to acquire information or not, along with a function $f: \Sigma \to A$ indicating that when the signal realization is $\sigma$, the agent takes action $f(\sigma)$.
We assume the principal is interested in a plan that includes information acquisition.
We say a plan is \emph{elicited} by a menu $T$ if \emph{(a)} the agent maximizes expected utility by following the plan instead of any other plan (incentive constraint); and \emph{(b)} the agent's expected utility is at least zero (participation constraint).
To recap, our goal is to solve the \emph{minimum payment problem}: given a plan, \emph{(1)} can it be elicited?
\emph{(2)} If so, what is the minimum expected payment to do so under limited liability (defined next)?

\paragraph{Limited liability.}
Formally, a menu $T$ satisfies \emph{limited liability} if for all $t \in T$ and all $\omega \in \Omega$, $t(\omega) \geq 0$. 
In other words, regardless of the contract selected and the outcome, the agent never makes a net payment to the principal.

\begin{example} \label{example:tv-producer}
Suppose a television company (principal) appoints a show producer (agent).
The company is only concerned with whether the show is a hit ($\omega = 1$) or not ($\omega = 0$).
The producer can choose to pay $\kappa > 0$ for a market research study (this is $S$); with 70\% probability (this is $q$), it will find that subscribers prefer shows with a woman lead ($S=w$), otherwise, it will find they do not ($S=m$).
The producer has a choice of two scripts, $a$ and $b$ (the actions).
If subscribers prefer woman leads, then the probability that $a$ is a hit is 80\%, i.e. $p_{aw} = (0.2,0.8)$.
Otherwise, the probability is only 40\%, i.e. $p_{am} = (0.6,0.4)$.
Similarly, the respective probabilities that $b$ is a hit are $p_{bw} = (0.5,0.5)$ and $p_{bm} = (0.7,0.3)$.
We see that regardless of the market research, show $a$ is more likely to be a hit than show $b$.
However, show $a$ is more costly, i.e. $c_a > c_a$.
For some other reasons like outreach or user-retention, perhaps the television company prefers a plan where the producer first conducts the research, and in the case $S=w$ produces script $a$, but in the case $S=m$ produces the cheaper script $b$.
While in this example the action is not completely hidden, the television company cannot observe whether the producer conducts research, so it cannot tell if she is following the plan.
Hence it may choose to treat the entire problem as Contracts with Information Acquisition.
Similarly, it can often be the case that actions are not completely hidden, but are impractical for the principal to verify. 
\end{example}

\paragraph{Special cases: information acquisition and contracts.}
Suppose we eliminate Step (4) of the game, where the agent privately selects an action.
After the agent chooses to acquire the signal $S$ or not and selects a contract from the menu, the observation is drawn from a distribution $p_S$ that only depends on $S$.
In this case, our model reduces precisely to \emph{information acquisition (IA)}.
The goal of the principal is simply to incent the agent to acquire the signal and report its realization.
The agent implicitly reports the realization in Step (3) by choosing a contract from the menu.
We will solve this special case first, in Section \ref{sec:info_acq}.

Now, instead, suppose we keep Step (4) and eliminate Step (2), where the agent decides whether or not to acquire the signal.
In this case, the principal proposes a menu, then the agent immediately selects a contract, then chooses an action $a^*$.
The event $W$ is drawn from a distribution $p_{a^*}$ that depends on $a^*$ alone.
In this case (studied in Section \ref{sec:contract-theory}), our model reduces precisely to a classical principal-agent contracts problem.

\subsection{Contract menus are proper scoring rules} \label{sec:scoring-rules}
The key observation underlying this work is that a menu of contracts can, without loss of generality, be represented as a \emph{proper scoring rule}, hence as a convex function.\footnote{Recall that $G$ is convex if for all $x,y$ and all $\alpha \in [0,1]$, $G(\alpha x + (1-\alpha)y) \leq \alpha G(x) + (1-\alpha) G(y)$.}

Used in other contexts to incent experts to make good predictions, a scoring rule\footnote{We disallow scores of $-\infty$ in this work, as we have required contracts to only pay off real numbers. This is without loss of generality, as paying $-\infty$ would violate limited liability in a rather extreme fashion. An implication is that, unlike in the most general scoring rule setting of e.g. \cite{gneiting2007strictly}, we only need to deal with convex functions that are subdifferentiable (see below), e.g. as in \cite{savage1971elicitation}.}~\citep{gneiting2007strictly} is a function $s: \Delta_{\Omega} \times \Omega \to \reals$.
When an expert makes a prediction $p \in \Delta_{\Omega}$ and the principal later observes an outcome $\omega \in \Omega$, the score is $s(p,\omega)$.
A scoring rule is \emph{(strictly) proper} if the agent (uniquely) maximizes expected score by reporting truthfully.
For example, the quadratic scoring rule $s(p,\omega) = 2p(\omega) - \sum_{\omega'} p(\omega')^2$ is strictly proper.

Observe that for fixed $p$, the function $s(p,\cdot)$ is analogous to a \emph{contract} $t(\cdot)$ in our setting.
A proper scoring rule $s$ can be represented as a menu of contracts $T = \{s(p,\cdot) : p \in \Delta_{\Omega}\}$.
Making a prediction $p$ is equivalent to selecting a contract $t$ from the menu, which implies a belief $p$ by revealed preference.
A well-known characterization~\citep{savage1971elicitation} states that a scoring rule $s$ is (strictly) proper if and only if there exists a (strictly) convex function $G: \Delta_{\Omega} \to \reals$ such that $s(p,\cdot)$ is a \emph{(strict) subtangent} of $G$ at $p$, where:
\begin{definition}[Expected payment, subtangent] \label{def:subtangent}
  Given a contract $t: \Omega \to \reals$, the \emph{expected payment function} is $\bar{t}: \Delta_{\Omega} \to \reals$ given by $\bar{t}(p) := \E_{W \sim p} t(W) = \sum_{\omega \in \Omega} p(\omega) t(\omega)$.
  We say $t$ is a \emph{(strict) subtangent} of a convex function $G: \Delta_{\Omega} \to \reals$ at a point $p$ if $\bar{t}(p) = G(p)$ and $\bar{t}(p')$ is (strictly) less than $G(p')$ for all $p' \neq p$.
\end{definition}
We may also use \emph{subtangent} to refer to any affine\footnote{A function is affine if it is linear plus a constant, i.e. is of the form $\hat{p} \mapsto \langle \hat{p}, v \rangle + \beta$ for some $v \in \reals^{\Omega}$ and some $\beta \in \reals$.} function that lies weakly below $G$ everywhere and is equal to $G$ at $p$.
We recall from convex analysis~\cite{hiriarturrut2001fundamentals}:
\begin{definition}[Subgradient, subdifferentiable]
  A vector $v$ is a \emph{subgradient} of a function $G$ at a point $p$ if, for all $p'$, we have $G(p') \geq G(p) + v \cdot (p' - p)$.
  A function $G$ is \emph{subdifferentiable} if it has at least one subgradient at every point in its domain.
\end{definition}
If $t$ is a subtangent contract of $G$ at $p$, then $\bar{t}$ can be expressed as
  \[ \bar{t}(p') = G(p) + v \cdot (p' - p) , \]
where $v$ is a subgradient of $G$ at $p$.
As subgradients are generalizations of gradients, $\bar{t}$ can be interpreted as a linear approximation to $G$ at $p$.
This is pictured in Figure \ref{fig:example-G-t}.

We recall that a convex $G$ is subdifferentiable if and only if it can be written as a pointwise maximum over a (possibly-infinite) set of affine functions, i.e. $G(p) = \max_{i \in \mathcal{I}} f_i(p)$.
(At each $p$, the gradient of any affine $f_i$ that achieves the max is a subgradient of $G$.)

\begin{figure}
\centering
\caption{Illustration of convexity, scoring rules, and contracts using Example \ref{example:tv-producer}. The horizontal axis is the probability that the show is a hit; e.g. if $S=w$ and the agent takes action $a$, the probability is $p_{aw}(1) = 0.8$. In black is a contract $t$ with $t(1) = 0.8$ (upper-right point), $t(0) = 0$ (lower-left point). The dotted line is the expected payment function $\bar{t}(p)$. In blue is the convex function $G$, which represents a menu. The contracts in the menu are derived from the subtangents of $G$. In particular, $\bar{t}$ is a subtangent of $G$ at $p_{aw}$, indicating that $t$ is an optimal contract for the agent to choose if $S=w$ and $a^*=a$. Optimality of $t$ follows because every other subtangent of $G$, say $\bar{t'}$, satisfies $\bar{t'}(p_{aw}) \leq \bar{t}(p_{aw})$, by convexity of $G$.}
\label{fig:example-G-t}
  \captionsetup{width=.9\linewidth,font=small}
\begin{tikzpicture} [scale=0.75]
\begin{axis}
    \addplot [mark=none,  black,   ultra thick, dotted] coordinates { (0,0) (1,0.8)};
\end{axis}
\draw[thick, blue, smooth] (0.5, 2.3)
      .. controls (2.1, 2.5) and (4.7, 3.2) .. (6, 5.5) ;
     \node [] at (5.3,5.2){ $G$}; 
     \node [] at (3,2){ $\bar{t}$};
     \node[] at (0.55, 1){\tiny $(0,t(0))$};
     \fill (0.55,0.5) circle (2pt);
     \node[] at (6.2, 4.5){\tiny $(1,t(1))$};
     \fill (6.2,5.15) circle (2pt);
     \fill (2.3,0) circle (2pt) ;
     \node[] at (2.1,0.3) {\tiny $p_{bm}$};
     \fill (2.85,0) circle (2pt) ;
     \node[] at (2.85,0.3) {\tiny $p_{am}$};
     \fill (3.4,0) circle (2pt) ;
     \node[] at (3.6,0.3) {\tiny $p_{bw}$};
     \fill (5.15,0) circle (2pt) ;
     \node[] at (5.15,0.3) {\tiny $p_{aw}$};
     
\end{tikzpicture}
\end{figure}
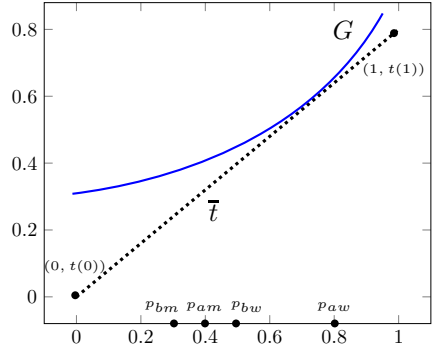

The same reasoning underlying the characterization applies in our setting (cf. \cite{frongillo2021general} for a very general version).
That is, we will show next that we can represent a menu $T$ as the set of subtangents of a subdifferentiable convex function $G: \P \to \reals$, where the domain $\P$ is defined as follows.
\begin{definition}[Convex hull, $\P$] \label{def:conv-hull-P}
  In a Contracts with Information Acquisition setting, define $\P = \text{convhull}(\{p_{a,\sigma} : a \in A, \sigma \in \Sigma\})$.
  Here, given a set of points $X \subseteq \reals^d$, their \emph{convex hull $\text{convhull}(X)$} is the set $\{\E_{x \sim r} x ~:~ r \in \Delta_X\}$.
\end{definition}
\begin{prop} \label{prop:menu-convex}
  In the setting of Contracts with Information Acquisition, without loss of generality, the menu $T$ is the set of subtangents of a subdifferentiable convex function $G: \P \to \reals$, which is the pointwise maximum of the respective expected payment functions.
\end{prop}
\begin{proof}
  Let $T$ be any menu and take the set $T' = \cup_{p \in \Delta_{\Omega}} \arg\max_{t \in T} \bar{t}(p)$.
  In other words, $T'$ contains, for each possible belief $p$ the agent may have about $W$ at the time of selecting a contract, all contracts that maximize expected utility according to $p$.
  Any contract in $T \setminus T'$ is strictly dominated, i.e. has strictly lower expected utility for every possible belief.
  So without loss of generality, the menu of contracts is $T'$.
  
  Let $G(p) = \max_{t \in T} \bar{t}(p)$.
  Note that $G$ is convex and subdifferentiable, because it is a pointwise maximum of linear functions.
  And $T'$ is its set of subtangents, as each $t \in T'$ achieves the maximum at some $p$ by definition of $T'$, and their expected payment functions all lie weakly below $G$ by definition of $G$.
\end{proof}

Given this result, we will often represent a menu $T$ by a subdifferentiable convex function $G: \P \to \mathbb{R}$, where the set of subtangents is assumed to be specified with $G$.
The only contracts not included in such a representation are strictly dominated contracts, which are never selected by the agent.
We observe that such a $G$ can be extended to all of $\Delta_{\Omega}$, because each affine function in the expression $G(p) = \max_{i \in \mathcal{I}} f_i(p)$ can be extended.
So we will sometimes refer to $G$ as being defined on the entire simplex.
One important subtlety is that one must include with $G$ the specification of which subtangents are selected, when there are multiple subgradients at a given point.

\paragraph{Useful facts.}
Given the scoring rule representation, we make a few more observations.
\begin{definition}[$M_G(\omega)$] \label{def:subtanG-minpay}
  Given a subdifferentiable convex $G: \Delta_{\Omega} \to \reals$, define the minimum payment at outcome $\omega$ to be $M_G(\omega) := \min_{t} t(\omega)$, where the minimum is over subtangent contracts of $G$.
  In other words,
  \[ M_G(\omega) = \min_{p \in \P} \left[ G(p) + v_p \cdot (\delta_{\omega} - p) \right] , \]
where $v_p$ is the given choice of subgradient of $G$ at $p$.
\end{definition}
The limited liability condition can now be expressed as $M_G(\omega) \geq 0 ~ (\forall \omega \in \Omega)$.
Note that limited liability requires that every contract in the menu has nonnegative payment for every possible observation.
Also, the following important fact about scoring rules carries over.
It says that $G(p)$ is the expected payment to an agent who believes $p$ and (optimally, truthfully) predicts it.
\begin{prop} \label{prop:expected-G}
  Given a menu of contracts $G$, and given that the signal is $S=\sigma$ and the agent plans to select action $a^*$, the optimal contract for the agent is the subtangent of $G$ at $p_{a^*,\sigma}$.
  Conditioned on $S$ and $a^*$, the expected payment to the agent if they choose the optimal contract is $G(p_{a^*,\sigma})$.
\end{prop}
\begin{proof}
  Under the givens, the conditional distribution of the outcome $W$ is $p_{a^*,\sigma}$.
  For any contract $t$, the expected payment is $\E_{W \sim p_{a^*,\sigma}} t(W) = \bar{t}(p_{a^*,\sigma})$.
  By definition, this value is less than or equal to $G(p_{a^*,\sigma})$, with equality if and only if $t$ is a subtangent of $G$ at $p_{a^*,\sigma}$.
\end{proof}

\section{ Information Acquisition } \label{sec:info_acq}
In this section, we consider the special case of our game where the agent only acquires information and does not take actions (equivalently, the agent's actions do not influence the true distribution on outcomes).
The principal's goal reduces to incenting the agent to acquire the signal and select a contract that reveals its' realization.
We call this the \emph{Information Acquisition (IA)} problem.
See Section \ref{sec:related-work} for discussion of related work on IA.

\subsection{Model and optimization problem}

In IA, the information structure consists of the distribution of the signal, $q \in \Delta_{\Sigma}$, and the set of posteriors on the outcome given each signal realization, $\{p_{\sigma} : \sigma \in \Sigma\}$.
We let the prior on the outcome be $p_0 = \E_{S \sim q} p_S$.
This can be realized as a special case of our general model by setting $c_a=0$ and $p_{a,\sigma}=p_{\sigma}$ and  $p_a=p_0$ for all $a \in A$ and $\sigma \in \Sigma$.
We assume that $p_0$ has full support on $\Omega$.
Otherwise, we can remove the zero-mass elements from $\Omega$ without loss.

The principal's minimum payment problem reduces to incentivizing information acquisition as cheaply as possible subject to limited liability.

\begin{prop} \label{prop:IA-opt-prob}
  In the Information Acquisition setting, the principal's minimum payment problem for incentivizing acquisition, given parameters $\kappa$, $q$, and $\{p_{\sigma} : \sigma \in \Sigma\}$, is solved by Program \ref{IA-main}.
\end{prop}
\begin{alignat}{2}
  \lplabel[IA-main]{(P1)} \text{} \quad \min_G ~ \E_{S \sim q}[G(p_{S})]  &           & \quad & \nonumber \\
  \text{s.t.} \quad  \text{  $G$ is subdifferentiable convex and } & &\quad&\nonumber \\
  \quad \E_{S \sim q}[G(p_{S})] - \kappa & \geq G(p_{0}) & \quad & \label{eqn:IA-incent-constr} \\
  M_G(\omega) & \geq 0,   & \quad & \forall \omega \in \Omega.  \label{eqn:IA-LL-constr}
\end{alignat}
\begin{proof}
  Recall by Proposition \ref{prop:menu-convex} that a menu can be represented WLOG by a subdifferentiable convex function $G: \P \to \reals$, where here $\P = \text{convhull}(\{p_{\sigma} : \sigma \in \Sigma\})$.
  Recall by Proposition \ref{prop:expected-G} that an agent with belief $p \in \P$ maximizes expected payment by selecting the contract $t$ that is a subtangent of $G$ at $p$, and her  expected payment in this case is $G(p)$.
  Therefore, the expected payment if the agent does not acquire the signal is $G(p_0)$.
  If she does acquire the signal, it is $\E_{S \sim q} G(p_S)$, and her  net utility in this case is $\E_{S \sim q} G(p_S) - \kappa$.
  The objective and the incentive constraint (\ref{eqn:IA-incent-constr}) follow immediately.
  The final constraint (\ref{eqn:IA-LL-constr}) is limited liability.
  The participation requirement is already implied in this setting by limited liability, as in particular constraint (\ref{eqn:IA-LL-constr}) implies $G(p_0) \geq 0$, so constraint (\ref{eqn:IA-incent-constr}) implies that the agent's net utility is nonnegative.
\end{proof}

We will consider \emph{nontrivial} settings.
\begin{definition}[Nontrivial IA] \label{def:IA-nontrivial}
  An Information Acquisition setting $\kappa$, $q$, $\{p_{\sigma} : \sigma \in \Sigma\}$ is \emph{nontrivial} if $\kappa > 0$ and there exists $\sigma$ with $p_{\sigma} \neq p_0$.
\end{definition}
If a setting has $\kappa = 0$, then the agent can costlessly acquire information and the principal's problem is trivial.
If a setting has $\kappa > 0$ but $p_{\sigma} = p_0$ for all $\sigma \in \Sigma$, then the signal is irrelevant to the observation and the principal cannot incentivize information acquisition by any means.

\subsection{Results}
We will give an optimal solution in closed form.
It will also be ``low detail'': it involves computing a convex function $\bar{G}^*$ that only depends on the prior $p_0$ and no other parameters; then scaling $\bar{G}^*$ by a constant that depends on just two numbers, the cost of information $\kappa$ and the ``value of information'' $\E_{S \sim q} \bar{G}^*(p_S)$.
Outside of the quantity $\E \bar{G}^*(p_S)$, the solution does not depend on the precise parameters $\{p_{\sigma} : \sigma \in \Sigma\}$ and $q \in \Delta_{\Sigma}$.
Geometrically, a nice property of the solution is that it is always of the following form.
\begin{definition}[Pointed polyhedral cone] \label{def:G-cone}
  We say a function $G: \Delta_{\Omega} \to \reals$ is a \emph{pointed polyhedral cone (PPC) at $p_0$} if $G$ can be written as the pointwise maximum over a finite number of affine functions, each of which passes through $(p_0,G(p_0))$.
\end{definition}
If $G$ is a PPC at $p_0$, then its epigraph (the set of points lying on or above the graph of $G$) is indeed polyhedral (an intersection of a finite number of closed halfspaces) and a cone (a set that, if it contains $x$, contains $\alpha x$ for all $\alpha \geq 0$), where the cone has been shifted so that its lowest point is $(p_0, G(p_0))$.
\begin{theorem} \label{thm:IIA-PPC}
  For any nontrivial IA setting $p_0,\kappa,q,\{p_\sigma ~:~ \sigma \in \Sigma\}$, there exists a PPC $G^*$ that optimally solves the principal's minimum payment problem.
  In particular, an optimal solution is as follows: define $\bar{G}^*(p) = \max_{\omega} \frac{p(\omega)}{p_0(\omega)}$, and let $G^* = \alpha \bar{G}^*$ where $\alpha = \kappa/\left(\E \bar{G}^*(p_S) - 1\right)$.
  Furthermore, this solution is ``low detail'' in that it only depends on $p_0,\kappa$, and the number $\E_{S \sim q} \bar{G}^*(p_S)$, but not otherwise on the details of $q$ or $\{p_{\sigma} ~:~ \sigma \in \Sigma\}$.
\end{theorem}
The proof is indirect and requires a number of steps.
We will present a \emph{maximization} problem, Program \ref{IA-dual}, whose objective function is the negative of that in \ref{IA-main}.
Remarkably, although their objectives are exactly opposed, the optimal solutions to \ref{IA-main} and \ref{IA-dual} will turn out to be simple rescalings of each other.
The intuition is that \ref{IA-main} involves ``squishing'' a convex function as flat as possible, while \ref{IA-dual} involves ``pulling'' it as tall as possible, and the same shape is optimal for both goals under their respective constraints.
After proving the equivalence, we produce a closed-form solution to \ref{IA-dual} in Theorem \ref{thm:IIA-PPC}.

We now define Program \ref{IA-dual}, with parameters $q$ and $\{p_{\sigma} : \sigma \in \Sigma\}$.
For clarity, we will always use $\bar{G}$ to represent a solution to \ref{IA-dual} and $G$ to represent a solution to \ref{IA-main}.

\begin{alignat}{2}
  \lplabel[IA-dual]{(P2)} \text{}   \quad \max_{\bar{G}} \E_{S \sim q}[\Bar{G}(p_{S})]  &           & \quad & \nonumber \\
  \text{s.t.} \quad  \text{  $\bar{G}$ is subdifferentiable convex and } & &\quad&\nonumber \\
  \quad \Bar{G}(p_0) & \leq 1 & \quad & \label{eqn:IA-dual-one-constr} \\
  \quad \E_{S \sim q}[\Bar{G}(p_{S})] & > 1 & \quad & \label{eqn:IA-dual-dumb-constr} \\
  M_{\bar{G}}(\omega) & \geq 0,   & \quad & \forall \omega \in \Omega.  \label{eqn:IA-dual-LL-constr}
\end{alignat}
The intuition of this optimization problem is to ``stretch'' $\bar{G}$ as tall as possible subject to having a low point (\ref{eqn:IA-dual-one-constr}) and limited liability (\ref{eqn:IA-dual-LL-constr}).
Limited liability ensures that $\bar{G}$ cannot be too tall, as then its subtangent planes would be very steep and would dip below $0$ at some corner of the simplex.

Constraint (\ref{eqn:IA-dual-dumb-constr}) does not affect the optimal solution.
It eliminates some trivial non-optimal feasible solutions.
By including Constraint (\ref{eqn:IA-dual-dumb-constr}), we obtain the following one-to-one relationship between feasible solutions of \ref{IA-dual} and \ref{IA-main}.

\begin{lemma} \label{lemma:IA-feasible-equiv}
  Let a nontrivial IA setting be given.
  Define $\phi(G) = \frac{G}{\E G(p_S) - \kappa}$ and define $\phi'(\bar{G}) = \frac{\kappa \bar{G}}{\E \bar{G}(p_S) - 1}$.
  Then $G$ is a feasible solution to \ref{IA-main} if and only if $\phi(G)$ is a feasible solution to \ref{IA-dual}; and $\bar{G}$ is feasible for \ref{IA-dual} if and only if $\phi'(\bar{G})$ is feasible for \ref{IA-main}; and $\phi^{-1} = \phi'$ on these feasible sets.
\end{lemma}
\begin{proof}
  We first show that if $G$ is feasible for \ref{IA-main}, then $\phi(G)$ is feasible for \ref{IA-dual}.
  We then show that if $\bar{G}$ is feasible for \ref{IA-dual}, then $\phi'(\bar{G})$ is feasible for \ref{IA-main}.
  We then show that $\phi$ and $\phi'$ are inverses on the feasible solution sets, which completes the proof.
  
  Let $G$ be feasible for \ref{IA-main}.
  First, we claim that $G(p_0) > 0$.
  By nontriviality and the incentive constraint (\ref{eqn:IA-incent-constr}), there exists $\sigma$ with $G(p_{\sigma}) > G(p_0)$.
  Consider the subtangent $t$ at $p_{\sigma}$ and write $p_0$ as a convex combination of $p_{\sigma}$ and $\{\delta_{\omega} : \omega \in \Omega\}$ with strictly positive weight on $p_{\sigma}$.
  (This is always possible; e.g. it follows from Lemma \ref{lemma:p-p0-interior}.)
  Because $\bar{t}(p_{\sigma}) > 0$ and $\bar{t}(\delta_{\omega}) \geq 0$ by limited liability, we obtain that $\bar{t}(p_0) > 0$.
  Because $t$ is a subtangent of $G$, $G(p_0) \geq \bar{t}(p_0) > 0$.
  
  So we have $G(p_0) > 0$, and the incentive constraint (\ref{eqn:IA-incent-constr}) gives $\E G(p_S) - \kappa \geq G(p_0) > 0$.
  For convenience, let $\alpha = \frac{1}{\E G(p_S) - \kappa}$, and let $\bar{G} = \phi(G) = \alpha G$.
  The point of the proof so far is that $\alpha$ is strictly positive and well-defined (i.e. the denominator is positive).
  Observe using the definition that $v$ is a subgradient of $G$ at $p$ if and only if $\alpha v$ is a subgradient of $\bar{G}$ at $p$.
  This implies that $\bar{G}$ is a subdifferentiable convex function.
  
  Now $\bar{G}(p_0) = \alpha G(p_0) \leq 1$ immediately from constraint (\ref{eqn:IA-incent-constr}) in \ref{IA-main}, showing that constraint (\ref{eqn:IA-dual-one-constr}) in \ref{IA-dual} is satisfied.
  Meanwhile, $\E \bar{G}(p_S) = \alpha \E G(p_S) > 1$ because $\kappa > 0$, giving constraint (\ref{eqn:IA-dual-dumb-constr}).
  
  Finally, for the limited liability constraint (\ref{eqn:IA-dual-LL-constr}), we use the definition of subtangents and subgradients to observe that if $t'$ is a subtangent of $\bar{G}$ at a point $p$, then $\frac{1}{\alpha} t'$ is a subtangent of $G$ at $p$.
  It follows immediately from Definition \ref{def:subtanG-minpay} (of $M_G$) that if $M_G(\omega) \geq 0$ then $M_{\bar{G}}(\omega) \geq 0$, completing the proof that $\bar{G}$ is feasible.
  
  For the second part of the proof, suppose $\bar{G}$ is feasible for \ref{IA-dual}.
  Let $\beta = \frac{\kappa}{\E \bar{G}(p_S) - 1}$, observing that $\beta \in (0,\infty)$ by nontriviality and constraint (\ref{eqn:IA-dual-dumb-constr}).
  Let $G(p) =\phi'(\bar{G})= \beta \bar{G}(p)$, a subdifferentiable convex function.
  Consider
  \begin{align*}
    \E G(p_S) - G(p_0)
    &=    \beta \left(\E \bar{G}(p_S) - \bar{G}(p_0)\right)  \\
    &\geq \beta \left(\E \bar{G}(p_S) - 1 \right)  & \text{using constraint (\ref{eqn:IA-dual-one-constr})} \\
    &=    \kappa ,
  \end{align*}
  so $G$ satisfies the incentive constraint (\ref{eqn:IA-incent-constr}).
  Now the limited liability argument is exactly the same as above, namely that $M_G(\omega) = \beta M_{\bar{G}}(\omega)$.
  So $G$ is feasible for \ref{IA-main}.
  
  Now we show that $\phi^{-1} = \phi'$ on the respective feasible sets.
  Let $G$ be given, feasible for \ref{IA-main}.
  Let $\bar{G} = \phi(G) = \alpha G$, where $\alpha = \frac{1}{\E G(p_S) - \kappa}$ for convenience..
  Then
  \begin{align*}
    \phi'(\bar{G})
    &= \frac{\kappa \bar{G}}{\E \bar{G}(p_S) - 1}  \\
    &= \frac{\kappa \alpha G}{\alpha \E G(p_S) - 1}  \\
    &= \frac{\kappa G}{\E G(p_S) - (1/\alpha)}  \\
    &= \frac{\kappa G}{\kappa}  \\
    &= G .
  \end{align*}
  This completes the proof.
\end{proof}


Now we show that the $\phi$ relationship preserves optimality.

\begin{lemma}\label{lemma:IA-opt-equiv}
  Let the program parameters $\kappa$, $q$, $\{p_{\sigma} : \sigma \in \Sigma\}$
  be fixed.
  Define $\phi(G) = \frac{G}{\E_{S \sim q}\left[G(p_S)\right] - \kappa}$.
  Then $G^*$ is an optimal solution to \ref{IA-main} if  and only if $\phi(G^*)$ is an optimal solution to \ref{IA-dual}.
\end{lemma}
\begin{proof}
  Let $G^*$ be an optimal solution to \ref{IA-main} and let $\bar{G}^* = \phi(G^*)$.
  Lemma \ref{lemma:IA-feasible-equiv} implies that $\bar{G}^*$ is feasible for \ref{IA-dual}.
  Now let $\bar{G}$ be any other feasible solution to \ref{IA-dual} and let $G = \phi^{-1}(\bar{G})$, a feasible solution to \ref{IA-main} by Lemma \ref{lemma:IA-feasible-equiv}.
  
  We show that $\bar{G}^*$ has at least as high an objective:
      \begin{align*}
          \E[\Bar{G}^*(p_S)] &= \frac{\E[G^*(p_S)]}{\E[G^*(p_S)]-\kappa}\\
                  &= 1+ \frac{\kappa}{\E[G^*(p_S)]-\kappa}\\
                  &\geq 1+ \frac{\kappa}{\E[G(p_S)]-\kappa} &\text{using optimality of $G^*$} \\
                  &= \E[\Bar{G}(p_S)]\\
      \end{align*}
  This shows that $\Bar{G^*}$ has at least as high an objective function as $\bar{G}$ in \ref{IA-dual}.
  Since this holds for all feasible $\bar{G}$, we obtain that $\bar{G}^*$ is optimal.
  
  For the converse, let $\Bar{G}^*$ be an optimal solution for \ref{IA-dual} and define $G^* = \phi^{-1}(\bar{G}^*) = \frac{\kappa \bar{G}^*}{\E[\bar{G}^*(p_S)] - 1}$.
  By Lemma \ref{lemma:IA-feasible-equiv}, $G^*$ is feasible for \ref{IA-main}.
  Let $G$ be any other feasible solution and $\bar{G} = \phi(G)$.
  \begin{align*}
          \E[G^*(p_S)] &= \frac{\kappa\E[\Bar{G}^*(p_S)]}{\E[\Bar{G}^*(p_S)]-1}\\
                  &= \kappa+\frac{\kappa}{\E[\Bar{G}^*(p_S)]-1}\\
                  &\leq \kappa+\frac{\kappa}{\E[\Bar{G}(p_S)]-1} & \text{using optimality of $\bar{G}^*$}\\
                  &= \E[G(p_S)] .
  \end{align*}
  Because $G^*$ has a weakly lower objective value than any other feasible solution, it is optimal.
\end{proof}

We now present an optimal solution to \ref{IA-dual}.
For each $\omega \in \Omega$, define $h^*_{\omega}(p)$ to be the linear function that interpolates the points $(p_0,1)$ and $\{(\delta_{\omega'}, 0) : \omega' \in \Omega, \omega' \neq \omega\}$.
In particular, that function is
\begin{align*}
  h^*_{\omega}(p) &= \frac{1}{p_0(\omega)} \langle \delta_{\omega}, p \rangle  \\
          &= \frac{p(\omega)}{p_0(\omega)} .
\end{align*}
(Recall that we have assumed $p_0$ has full support.)
We see that $h^*_{\omega}(\delta_{\omega'}) = 0$ if $\omega' \neq \omega$, and that $h^*_{\omega}(p_0) = 1$, as required.
Now, let
  \[ \Bar{G}^*(p) = \max_{\omega} h_{\omega}^*(p) = \max_{\omega} \frac{p(\omega)}{p_0(\omega)} . \]
We observe that $\bar{G}^*$ is a polyhedral pointed cone (PPC) at $p_0$: it is the maximum over a finite number of hyperplanes, all of which pass through the point $(p_0,1)$.

\begin{theorem} \label{thm:IA-opt-dual}
  For any fixed prior $p_0$, the function $\Bar{G}^*(p) = \max_{\omega} \frac{p(\omega)}{p_0(\omega)}$ is an optimal solution to Program \ref{IA-dual}, regardless of the other parameters of the setting.
\end{theorem}
\begin{proof}
  For feasibility: we observe that $\bar{G}^*(p_0) = 1$.
  We also have $\E \bar{G}^*(p_S) > \bar{G}^*(p_0) = 1$ by Jensen's inequality, which is strict by nontriviality of the IA setting and construction of $\bar{G}^*$.\bo{TODO come back and prove that in more detail.}
  Finally, we have $M_{\bar{G}^*}(\omega) \geq 0$ for all $\omega$ because $\bar{G}^*$ is the pointwise maximum over the subtangents $\{h^*_{\omega} : \omega \in \Omega\}$, each of which is nonnegative on $\Delta_{\Omega}$.
  
  For optimality: Let $\bar{G}$ be any other feasible solution to Program \ref{IA-dual}.
  We show that $\bar{G}^* \geq \bar{G}$ pointwise, which implies its objective value is weakly higher.
  
  Fix any $p \in \Delta_{\Omega}$.
  By Lemma \ref{lemma:p-p0-interior}, for some $\omega$ and some probability distribution consisting of the numbers $\beta, \{\beta_{\omega'} : \omega' \neq \omega\}$, we can write $p_0 = \beta p + \sum_{\omega' \neq \omega} \beta_{\omega'} \delta_{\omega'}$.
  Furthermore, $\beta > 0$.
  
  Because $\bar{G}$ is subdifferentiable, there exists an affine function $h$ that is a subtangent of $G$ at $p$, in particular $h(p) = G(p)$ and $h(p_0) \leq G(p_0)$.
  We have by feasibility of $\bar{G}$ that $h(p_0) \leq 1$ and $h(\delta_{\omega'}) \geq 0$ for all $\omega'$.
  Using these two inequalities:
  \begin{align*}
    1 &\geq h(p_0)  \\
      &=    h\left(\beta p + \sum_{\omega' \neq \omega} \beta_{\omega}' \delta_{\omega'}\right)  \\
      &=    \beta h(p) + \sum_{\omega' \neq \omega} \beta_{\omega'} h(\delta_{\omega'})  \\
      &\geq \beta h(p) \\
      &=    \beta \bar{G}(p).
  \end{align*}
  We conclude $\bar{G}(p) \leq \frac{1}{\beta}$.

  Meanwhile, by construction of $\bar{G}^*$,
  \begin{align*}
    1 &= h^*_{\omega}(p_0)  \\
      &= \beta h^*_{\omega}(p) + \sum_{\omega' \neq \omega} \beta_{\omega'} h^*_{\omega}(\delta_{\omega'})  \\
      &= \beta h^*_{\omega}(p)  \\
      &\leq \beta \bar{G}^*(p) .
  \end{align*}
  We conclude $\bar{G}^*(p) \geq \frac{1}{\beta} \geq \bar{G}(p)$, as desired.
\end{proof}

We can now prove our main result, a solution to the IA problem with limited liability.
\begin{proof}[Proof of Theorem \ref{thm:IIA-PPC}]
  By Theorem \ref{thm:IA-opt-dual}, given $p_0$, $\bar{G}^*$ is an optimal solution to Program \ref{IA-dual} regardless of $\kappa$ or $\{p_{\sigma} : \sigma \in \Sigma\}$.
  By Lemma \ref{lemma:IA-opt-equiv}, the function $G^* = \phi^{-1}(\bar{G}^*)$ is therefore optimal for Program \ref{IA-dual} regardless of $\{p_{\sigma} : \sigma \in \Sigma\}$.
  As $\bar{G}^*$ is a polyhedral pointed cone, and $G^*$ is a scaling of $\bar{G}^*$, it is also a PPC.
\end{proof}

\begin{remark} \label{remark:IA-sol-indicators}
  Our optimal solution $G^*$ to the IA problem with limited liability can be viewed as a menu of $|\Omega|$ ``indicator'' contracts $t_{\omega}$, each of which pays off a positive amount if $W=\omega$ and zero otherwise.
  Each indicator $t_{\omega}$ has been scaled so that the expected utility for choosing it if the agent does not observe $S$ is a constant, $G^*(p_0) = \bar{t}_{\omega}(p_0)$. 
  (We obtain this because $G^*$ is a rescaling of $\bar{G}^*$, whose subtangents $h_{\omega}^*$ are indicators.)
\end{remark}

\begin{remark} \label{remark:how-we-IA}
  In order to eventually discover the closed-form solution for $\bar{G}^*$, we first found a series of transformations that take a generic feasible solution $\bar{G}$ and transform it into a PPC at $p_0$, with the limited liability (\ref{eqn:IA-dual-LL-constr}) and upper-bound (\ref{eqn:IA-dual-one-constr}) constraints tight.
  Each transformation preserves feasbility and can only improve the objective.
  These transformations are illustrated in Appendix \ref{app:IA-how-opt}.
  These results gave a restricted enough space that we were able to guess-and-check the optimal solution by hypothesizing the steepest possible subtangents that satisfied limited liability and $h_{\omega}^*(p_0) \leq 1$.
\end{remark}

\section{Contracts for Hidden Actions} \label{sec:contract-theory}
In this section, we consider the special case of the problem without information acquisition.
This reduces to a classic contract theory problem with hidden actions and limited liability, which has been addressed classically.
However, to our knowledge, the proper scoring rule perspective is novel.

Classically there is no need to offer more than one contract.
However, we will see that there is no loss in offering an entire menu of contracts based on a scoring rule, and we show when and how an optimal menu can contain other contracts.
This may be useful for more general settings.

One interpretation of the scoring-rule solution is as follows: rather than directing the agent to take a particular action, the principal simply asks the agent to make a prediction about $W$ and offers a proper scoring rule reward.
The agent finds that they can optimize their forecast accuracy by \emph{changing} the distribution of $W$, i.e. taking the action that was desired by the principal in the first place.
Furthermore, the improvement in forecast accuracy (i.e. increase in payment) is worth the cost the agent incurs in taking the action.

\subsection{Model and convexified cost curve} \label{sec:cost-curve}

The classic contracts setting can be achieved from our general model by removing the impact of the signal, i.e. setting $\kappa=0$ and $p_{a,\sigma}=p_{a}$ and $f(\sigma) = a^*$, for all $a \in A$ and $\sigma \in \Sigma$.
In this case, the principal's problem reduces to incentivizing a fixed action $a^*$.

We note that the classic contracts setting generally takes $\Omega$ to be a set of numerical outcomes denoting utility for the principal.
However, we do not make this assumption, and assume that the principal's preferences are encoded by the desired plan or action $a^*$.
With a solution to our problem, one can address those settings by iterating over the finitely many actions, solving the minimum payment problem for each, and then choosing the action that optimizes net utility.
(This is also the algorithm described in \citet{dutting2019simple} for the utility optimization problem.)

\begin{prop} \label{prop:contracts-program}
  In the contracts setting, the principal's minimum payment problem for incentivizing an action $a^*$, given parameters $\{c_a : a \in A\}, \{p_a : a \in A\}$, is solved by Program \ref{contracts-main}.
\end{prop}

\begin{alignat}{2}
                          \text{}   \quad \min_{G} G(p_{a^*})  &           & \quad & \nonumber \\
  \lplabel[contracts-main]{(P3)}\text{s.t.} \quad  \text{  G is subdifferentiable convex and } & &\quad&\nonumber \\
                                    \quad G(p_{a^*})- c_{a^*} & \geq G(p_a)-c_a & \quad & \forall a \in A  \label{eqn:contracts-incent-constr} \\
                                    \quad G(p_{a^*}) - c_{a^*} & \geq 0 & \quad & \quad  \label{eqn:contracts-partic-constr} \\
                                        M_G(\omega) & \geq 0,   & \quad & \forall \omega \in \Omega.  \label{eqn:contracts-LL-constr}
\end{alignat}
\begin{proof}
  By Proposition \ref{prop:menu-convex}, a menu can be represented WLOG by a subdifferentiable convex function $G$.
  By Proposition \ref{prop:expected-G}, an agent who intends to take action $a$, resulting in a belief $p_a$ about the distribution of $W$, maximizes expected payment by selecting the contract $t$ that is a subtangent of $G$ at $p_a$, and her  expected payment in this case is $G(p_a)$.
  The objective is thus to minimize expected payment when the agent takes $a^*$; constraint (\ref{eqn:contracts-incent-constr}) ensures the agent prefers taking action $a^*$ and selecting the optimal corresponding contract to taking any other action $a$; constraint (\ref{eqn:contracts-partic-constr}) ensures participation; and constraint (\ref{eqn:contracts-LL-constr}) is limited liability.
\end{proof}

Recall that $\P$ is the convex hull of the possible posterior beliefs, i.e. here $\P = \text{convhull}(\{p_a : a \in A\})$.
\begin{definition}[Convexified cost curve] \label{def:cost-curve}
  Given a contracts setting defined by $X = \{(p_a, c_a) : a \in A\}$, the \emph{convexified cost curve} is the function $c: \P \to \reals$ whose graph is the minimum of the convex hull of $X$, i.e.
  $c(p) = \min \{\E_{a \sim \lambda} c_a ~:~ p = \E_{a\sim \lambda} p_a, ~ \lambda \in \Delta_A\}$.
\end{definition}
The convexified cost curve can equivalently be defined as the \emph{lower convex envelope} of the points in $X$, i.e. as the pointwise maximum of all affine functions that lie below all these points.

The intuition is that, given a probability distribution $p \in \P$ over the outcomes, $c(p)$ represents the lowest possible cost to the agent to implement $p$ via a randomized choice of actions.
In general, if we can write $p = \sum_a \lambda(a) p_a$ for some probability distribution $\lambda$, then the agent can cause $W \sim p$ by first picking $a$ from $\lambda$, then performing action $a$.
The cost for doing so is $\E_{a\sim \lambda} c_a$; and $c(p)$ minimizes this cost over all such possible $\lambda$.
Thus, $c$ ``convexifies'' the agent's action space.
We will formalize the intuition that, if $c_a > c(p_a)$, then the principal cannot incentivize action $a$.
Intuitively, rather than playing $a$, the agent could more cheaply draw an action from a distribution $\lambda$ and perform that action.
This would result in the same distribution $p_a$ over outcomes $\Omega$ as performing that action, so the expected payment would be the same but the cost would be smaller.

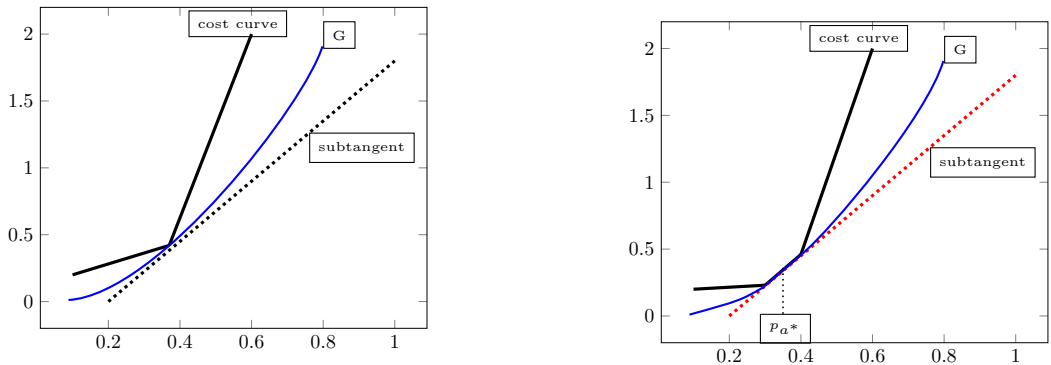
\begin{figure}[H]
\centering
\caption{The convexified cost curve; strict and weak elicitability.}
\label{fig:strict-weak-elicit}
\begin{subfigure}{.5\textwidth}
  \centering
  \captionsetup{width=.9\linewidth,font=small}
\begin{tikzpicture} [scale=0.75]
\begin{axis}
    \addplot [mark=none,  black,   ultra thick] coordinates { (0.1,0.2) (0.37, 0.42)} ;
    \addplot [mark=none,  black,   ultra thick] coordinates { (0.37, 0.42) (0.6,2)};
    \addplot [mark=none,  black,   ultra thick, dotted] coordinates { (0.2,0) (1,1.8)};
    
\end{axis}
\draw[thick, blue, smooth] (0.5, 0.5)
      .. controls (2.1, 0.6) and (4.7, 4) .. (5, 5) ;
     \node [draw] at (5.3,5.2){\tiny G}; 
     \node [draw] at (3.5,5.4){\tiny cost curve}; 
     \node [draw] at (5.7,3.2){\tiny subtangent};
\end{tikzpicture}
\caption{Here $G$ is ``sandwiched'' between the cost curve and subtangent of cost curve at the desired action. $G$, or an upward shift (to satisfy limited liability), is a feasible menu. Note that $G$, with the pictured subtangent contract, strictly elicits the action at the pictured vertex of the cost curve.}
\label{fig:strict-sandwich}
\end{subfigure}%
\begin{subfigure}{.5\textwidth}
  \centering
  \captionsetup{width=.9\linewidth,font=small}
\begin{tikzpicture} [scale=0.75]
\begin{axis}
    \addplot [mark=none,  black,   ultra thick] coordinates { (0.1,0.2) (0.3, 0.23)} ;
    \addplot [mark=none,  black,   ultra thick] coordinates { (0.3, 0.23) (0.4,0.46)};
    \addplot [mark=none,  black,   ultra thick] coordinates { (0.4, 0.46) (0.6,2)};
    \addplot [mark=none,  red,   ultra thick, dotted] coordinates { (0.2,0) (1,1.8)};
    \addplot [mark=none,  black, thick ,dotted] coordinates { (0.35,0.36) (0.35,0)};
\end{axis}
\draw[thick, blue, smooth] (0.5, 0.5) .. controls (1.4, 0.75) and (1.4,0.75) .. (1.8, 0.98) ;
\draw[thick, blue, smooth] (1.8, 0.98) .. controls (2.15, 1.28) and (2.15,1.28) .. (2.5, 1.58) ;
\draw[thick, blue, smooth] (2.5, 1.58) .. controls (3, 2) and (4.7, 4) .. (5, 5) ;
     \node [draw] at (5.3,5.2){\tiny G}; 
     \node [draw] at (3.5,5.4){\tiny cost curve}; 
     \node [draw] at (5.7,3.2){\tiny subtangent};
     \node[draw] at (2.2,0.25){\tiny $p_{a^*}$};
\end{tikzpicture}
\caption{A case where $a^*$ is elicitable, but not strictly elicitable. Any menu $G$ eliciting $a^*$ must have the dotted line as a subtangent, up to a vertical shift. This subtangent contract will weakly incentivize the other actions at the neighboring vertices of the cost curve.}
\label{fig:weak-sandwich}
\end{subfigure}%
\end{figure}

\subsection{Results}
We first characterize feasible solutions to the minimum payment problem, then optimal ones.
Actually, we begin by characterizing feasible solutions in the absence of the limited liability constraint.

\begin{definition}[Elicits, elicitable]\label{def:elicit}
  We say that a menu of contracts, presented as a subdifferentiable convex $G$, \emph{elicits} action $a^*$ if $G$ satisfies the incentive (\ref{eqn:contracts-incent-constr}) and participation (\ref{eqn:contracts-partic-constr}) constraints in Program \ref{contracts-main}.
  If such a $G$ exists, we say $a^*$ is \emph{elicitable}.
\end{definition}

All results and proofs appear in Appendix \ref{app:contracts}.
The intuition for this characterization was described above.
\begin{corollary*}[Corollary \ref{cor:contracts-elicitable}]
  An action $a^*$ is elicitable if and only if $c(p_{a^*}) = c_{a^*}$, i.e. if the point $(p_{a^*},c_{a^*})$ lies on the lower boundary of the convex hull of $\{(p_a,c_a) : a \in A\}$.\footnote{We note that this is similar to Proposition 2 of Appendix A2 in \citet{dutting2019simple} paper.}
\end{corollary*}

Building on Corollary \ref{cor:contracts-elicitable}, we give the following algorithm for computing an optimal solution to the classic contracts problem.
The main idea is to find an optimal subtangent of the convexified cost curve $c$ at the point $p_{a^*}$, then shift it to satisfy participation and limited liability.

\begin{prop*}[Proposition \ref{prop:contracts-minimal-opt}]
  If $a^*$ is elicitable, then Algorithm \ref{alg:contracts-minimal} computes an optimal solution to Program \ref{contracts-main} consisting of a single contract.
\end{prop*}

\begin{algorithm}
  \caption{Computing a a minimal optimal contracts menu $T$.}
  \label{alg:contracts-minimal}
  Given a contracts problem, let $c$ be the convexified cost curve and $a^*$ the desired elicitable action. \\
  Define $V_{a^*} = \arg\max_{v \in \partial c(p_{a^*})} \min_{\omega}   v(\omega) - v \cdot p_{a^*} .$ \tcp*{the optimal subgradients of $c$ at $p_{a^*}$}
  Let $v \in V_{a^*}$. \\
  Define $\beta = -\left( c(p_{a^*}) + \min_{\omega} v(\omega) - v \cdot p_{a^*} \right)$. \tcp*{the shift required for limited liability}
  Let $t(\omega) = c(p_{a^*}) + v \cdot (\delta_{\omega} - p_{a^*}) + \max\{0, \beta\}$. \\
  Let $T = \{t\}$ and $G = \bar{t}$. \tcp*{$\bar{t}$ is the expected payment function, Definition \ref{def:subtangent}}
\end{algorithm}

\begin{remark} \label{remark:contracts-LP}
  Finding a member of $V_{a^*}$ is a linear programming problem, as $\partial c(p_{a^*})$ is a closed convex polytope.
  However, it is not fundamentally different than a standard linear program for the minimum payment problem in the literature (\emph{e.g.} mentioned in \citet{dutting2019simple}, Section 2.).
  So while we believe our results add useful geometric intuition, they do not appear to enable faster algorithms.
\end{remark}

Additionally, one can generally add a number of additional contracts to the menu without compromising optimality.
We can view this process as taking an optimal subtangent, an affine $G$, and ``convexifying'' it further by adding other contracts that do not violate limited liability and do not lie above $c$ (after shifting down by the appropriate $\beta$).
Sometimes, one can substitute suboptimal subtangents of $c$ at $p_{a^*}$ instead as well, because participation rather than limited liability is the binding constraint.
The result is Algorithm \ref{alg:contracts-full}, given in Appendix \ref{app:contracts}.

\begin{prop*}[Proposition \ref{prop:contracts-all-opt}]
  If $a^*$ is elicitable, then every optimal solution to Program \ref{contracts-main} is computed by Algorithm \ref{alg:contracts-full} (given in Appendix \ref{app:contracts}), for some choices of tiebreakers.
\end{prop*}

\subsection{Strict elicitability}
It can happen that an action $a^*$ is elicitable, but there does not exist any menu that entices the agent to \emph{strictly} prefer $a^*$ (Figure \ref{fig:strict-weak-elicit}).
For completeness, we give a geometric characterization of strict elicitability and extend Algorithm \ref{alg:contracts-minimal} to produce menus that are arbitrarily close to optimal while still providing strict incentives.

\begin{definition} \label{def:strict-elicit}
  We say that a menu of contracts, presented as a subdifferentiable convex $G$, \emph{strictly elicits} action $a^*$ if $G$ satisfies the participation constraint (\ref{eqn:contracts-partic-constr}) in Program \ref{contracts-main} and satisfies the incentive constraint (\ref{eqn:contracts-incent-constr}) with strict inequality for all $a \neq a^*$.
  If such a $G$ exists, we say $a^*$ is \emph{strictly elicitable}.
  
  Related, we say a point $(p_{a^*}, c_{a^*})$ is a \emph{lower vertex} of $\mathrm{convhull}(\{(p_a,c_a) : a \in A\})$ if, for any $\lambda \in \Delta_{A \setminus \{a^*\}}$ such that $\E_{a \sim \lambda} p_a = p_{a^*}$, $\E_{a \sim \lambda} c_a > c_{a^*}$.
\end{definition}
In other words, the point $(p_{a^*},c_{a^*})$ is a lower vertex if it is a vertex (i.e. extreme point) of the convex hull, and additionally, it lies on the ``bottom'', i.e. no elements of the convex hull are below it.

We summarize the results here and refer the reader to Appendix \ref{app:contracts-strict} for details, algorithm definition, and proofs.

\begin{corollary*}[Corollary \ref{cor:strict-char}]
  An action $a^*$ is strictly elicitable if and only if $(p_{a^*}, c_{a^*})$ is a lower vertex of $\mathrm{convhull}(\{(p_a,c_a) : a \in A\})$.
\end{corollary*}

\begin{prop*}[Proposition \ref{prop:contracts-compute-strict}]
  In a contracts setting, if $a^*$ is strictly elicitable and $\epsilon > 0$, then Algorithm \ref{alg:contracts-strict} in Appendix \ref{app:contracts-strict} computes a menu $G$ that is feasible, that is $\epsilon$-optimal, and that strictly elicits $a^*$.
\end{prop*}
The idea behind Algorithm \ref{alg:contracts-strict} is straightforward.
We compute an optimal single contract $t_1$ as in Algorithm \ref{alg:contracts-minimal}.
Then, we use a \emph{strict subgradient} of $c({p_a^*})$, which must exist, to construct a contract $t_2$ that strictly elicits $a^*$.
That is, $t_2$ is a different, possibly suboptimal subtangent of $c(p_{a^*})$.
Now an arbitrary positive convex combination of $t_1$ and $t_2$ continues to strictly elicit $a^*$, and can be arbitrarily close to optimal.

\maneesha{Y did it skip figures 2, 4,6,8? \bo{dunno}}

\section{The General Problem} \label{sec:general}
Here, we address our general problem of Contracts with Information Acquisition.
Recall that the principal wishes to implement a plan $f: \Sigma \to A$ as cheaply as possible under limited liability.
That is, the principal must incentivize the agent to choose to acquire the signal $S$ and, for each signal realization $\sigma$, take the action $a=f(\sigma)$.

Recall that if the principal succeeds, i.e. the agent adopts the plan of acquiring $S$ and taking action $f(S)$, then after observing realization $S=\sigma$, the agent's belief about $W$ will be $p_{f(\sigma),\sigma}$.


\begin{prop} \label{prop:general-program}
  In the general problem of Contracts with Information Acquisition, the minimum payment problem of the principal is this: given the priors $q \in \Delta_{\Sigma}$ and $\{p_{a,\sigma} \in \Delta_{\Omega} : a \in A, \sigma \in \Sigma\}$, design a function $G: \Delta_{\Omega} \to \reals$ to solve Program \ref{general-main}.
\end{prop}

\begin{alignat}{2}
  \text{}   \quad \min_{G}\E_{S \sim q}[G(p_{f(S),S})]  &           & \quad & \nonumber \\
  \lplabel[general-main]{(P4)}\text{s.t.} \quad  \text{  $G$ is subdifferentiable convex and } & &\quad&\nonumber \\
    \quad \E_{S \sim q}[G(p_{f(S),S}) - c_{f(S)}] - \kappa & \geq G(p_a) - c_a & \quad &\forall a \in A  \label{eqn:general-nosig-constr} \\
    \quad G(p_{f(\sigma),\sigma})- c_{f(\sigma)} &  \geq G(p_{a,\sigma}) - c_{a} & \quad &\forall a \in A , \forall \sigma \in \Sigma \label{eqn:general-conditional-constr} \\
    \quad \E_{S \sim q}[G(p_{f(S),S}) - c_{f(S)}] -\kappa & \geq 0 & \quad &\quad  \label{eqn:general-partic-constr} \\
    M_G(\omega) & \geq 0,   & \quad & \forall \omega \in \Omega. \label{eqn:general-LL-constr}
\end{alignat}

\begin{proof}
  By Proposition \ref{prop:menu-convex}, a menu can be represented WLOG by a subdifferentiable convex $G$.
  By Proposition \ref{prop:expected-G},
  the agent's expected payment when $S=\sigma$ and $a^* = a$ is $G(p_{a\sigma})$.
  Hence, the objective is the principal's payment when the agent acquires information and follows the plan.
  The agent's utility for following the plan is therefore $\E_{S \sim q} \left[ G(p_{f(S),S}) - c_{f(S)}\right] - \kappa$.
  
  On the other hand, the agent has the following alternatives available.
  First, she can not acquire information and simply pick an action $a$, for a net expected utility of $G(p_a) - c_a$.
  Constraint (\ref{eqn:general-nosig-constr}) is the requirement that following the plan is preferred, for all $a$.
  Second, she can acquire information, but choose different actions than the plan specifies.
  In particular, conditioned on a realization $S=\sigma$, the expected utility for following the plan is $G(p_{f(\sigma),\sigma}) - c_{f(\sigma)}$, while the expected utility for choosing action $a$ is $G(p_{a,\sigma}) - c_a$.
  Constraint (\ref{eqn:general-conditional-constr}) is the requirement that the agent prefer following the plan, for all $\sigma$ and $a$.
  Next, constraint (\ref{eqn:general-partic-constr}) is the participation constraint, and (\ref{eqn:general-LL-constr}) is limited liability.
\end{proof}


\subsection{A linear program}
While we do not give a closed-form or simply-structured solution, here we give a polynomial time algorithm to solve Program \ref{general-main}. 
We hope that future work can identify more useful structure in the solution.

The key idea is that a piecewise linear $G$ suffices to solve \ref{general-main}.
We observe that the only relevant values of $G$ in Program \ref{general-main} are $G(p_{a,\sigma})$ for $a \in A, \sigma \in \Sigma$ along with $G(p_a)$ for $a \in A$.
Therefore, it turns out, one can take WLOG a piecewise linear $G$ with only $(|\Sigma|+1) |A|$ pieces, each of which is a subtangent contract at one of these beliefs.
For convenience, let us define $\bar{\Sigma} = \Sigma \cup \bot$ where $\bot$ is a dummy signal and $p_{a,\bot} = p_a$.
Now note that each subtangent contract at $p_{a\sigma}$, for some $a \in A,\sigma \in \bar{\Sigma}$, is an affine function of the form $h_{a,\sigma}(p) = x_{a,\sigma} \cdot p - y_{a,\sigma}$ for some $x_{a,\sigma} \in \reals^{\Omega}$ and some $y_{a,\sigma} \in \reals$.
The variables of the Program \ref{general-dual} are $\{(x_{a,\sigma},y_{a,\sigma}) : a \in A, \sigma \in \bar{\Sigma}\}$.
For presentation, we use the shorthand $h_{a,\sigma}$ to denote the associated affine function defined above.

\begin{alignat}{2}
                          \text{}   \quad \min_{x_{a,\sigma}, y_{a,\sigma} \mid a\in A, \sigma \in \bar{\Sigma}} \sum_{ \sigma \in \Sigma} q(\sigma) \cdot h_{f(\sigma),\sigma}(p_{f(\sigma),\sigma})  &           & \quad & \nonumber \\
  \lplabel[general-dual]{(P5)} \text{s.t.}
                        \quad \sum_{\sigma \in \Sigma} q(\sigma)( h_{f(\sigma),\sigma}(p_{f(\sigma),\sigma}) - c_{f(\sigma)}) - \kappa & \geq x_a \cdot p_a - y_a - c_{a} & \quad & \forall a \in A \label{eqn:general-lp-nosig-constr}\\
                        \quad h_{f(\sigma),\sigma}(p_{f(\sigma),\sigma}) - c_{f(\sigma)} & \geq  h_{a,\sigma}(p_{a,\sigma})  - c_{a}  & \quad & \forall a \in A , \forall \sigma \in \Sigma \label{eqn:general-lp-conditional-constr} \\
                       \quad \sum_{\sigma \in \Sigma} q(\sigma) ( h_{f(\sigma),\sigma}(p_{f(\sigma),\sigma}) - c_{f(\sigma)}) - \kappa & \geq  0  & \quad & \quad \label{eqn:general-lp-partic-constr} \\
                        \quad h_{a,\sigma}(\delta_{\omega})  & \geq  0  & \quad & \forall \omega \in \Omega,\forall a \in A,\forall \sigma \in \bar{\Sigma} \label{eqn:general-lp-LL-constr} \\
                         h_{a,\sigma}(p_{a,\sigma}) & \geq  h_{a',\sigma'}(p_{a,\sigma})  & & \forall a,a' \in A , \forall \sigma,\sigma' \in \bar{\Sigma}  \label{eqn:general-lp-convex-constr}
\end{alignat}
\begin{lemma} \label{lemma:general-equiv}
 Suppose Program \ref{general-main} is feasible. Then Program \ref{general-dual} is also feasible and any optimal solution of Program \ref{general-dual} is also an optimal solution of Program \ref{general-main}.
\end{lemma}
\begin{proof}
  To show the first claim, let $G$ be a feasible solution of Program \ref{general-main}. Define  $x_{a,\sigma} \in \partial G(p_{a,\sigma})$ to be the associated subgradient of $G$ at $p_{a,\sigma}$.
  Let $y_{a,\sigma} = x_{a,\sigma} \cdot p_{a,\sigma} - G(p_{a,\sigma})$.
  In other words, the affine function $h : p \mapsto x_{a,\sigma} \cdot p - y_{a,\sigma}$ is a subtangent of $G$ at $p_{a,\sigma}$, i.e. lies everywhere weakly below $G$ and equals $G$ at $p_{a,\sigma}$.
  Constraint (\ref{eqn:general-lp-convex-constr}) follows immediately from convexity of $G$.
  \bo{I think that is an acceptable level of explanation.}
   Similarly, every other constraint in Program \ref{general-dual} is a direct translation of the corresponding constraint in \ref{general-main}, where $G(p_{a,\sigma})$ has been exchanged for $h_{a,\sigma}(p_{a,\sigma}) = x_{a,\sigma} \cdot p_{a,\sigma} - y_{a,\sigma}$.
   The constraints only depend on the values of $G$ at those points, so they are unchanged.
   The only exception is the limited liability constraint, which depends on the subtangents of $G$.
   As our new program consists of a collection of some of the subtangents of $G$, and $G$ is feasible, all of our subtangents also satisfy limited liability.
   
%

  Now further suppose that $G$ was optimal for Program \ref{general-main}.
  Then the function $G'$ implicitly defined by $G'(p) = \max_{a,\sigma} h_{a,\sigma}(p) =  \max_{a,\sigma} x_{a,\sigma} \cdot p - y_{a,\sigma}$ is also optimal for Program \ref{general-main}.
  By the above arguments, it follows that $G'$ is feasible and has the same objective value as $G$.
  So the optimal value of Program \ref{general-main} is at most as large as that of \ref{general-dual}.

  On the other hand, given \emph{any} solution to Program \ref{general-dual}, we can define $G'$ as above.
  We observe that, by the same arguments as above,  $G'$ is also feasible for Program \ref{general-main} and its objective value is the same as the objective in \ref{general-dual}.
  It follows that the optimal value of Program \ref{general-main} is at least as large as that of \ref{general-dual}.
  So their optimal objective values are equal, and in particular any optimal solution to \ref{general-dual} yields an optimal solution to \ref{general-main} by defining $G'$.

\end{proof}

\begin{theorem} \label{thm:general-alg}
 For fixed action space, outcomes and signal realizations, Program \ref{general-main} is solvable in polynomial time.
\end{theorem}
\begin{proof}
  By Lemma \ref{lemma:general-equiv}, it suffices to solve Program \ref{general-dual} to obtain an optimal solution.
   We can see that \ref{general-dual} is a linear program with $2|A|(|\Sigma|+1)$ variables and $|A|+|A||\Sigma|+ 1 + |\Omega||A|(|\Sigma|+1) +|\Omega||A|^2(|\Sigma|+1)^2 $ constraints, all of them polynomial in given parameters. 
\end{proof}

\subsection{Efficiency and necessary conditions} \label{subsec:general-efficiency}
While we have produced a polynomial-time algorithm, we would like a better understanding of the structure of the solution as well as any computational speedups available.
In this section, we describe some necessary properties of an optimal solution, utilizing insights from the IA and hidden action settings.

We first observe a necessary condition for elicitability of certain plans. 
It extends the relationship of elicitability with the convexified cost curve defined in \ref{sec:cost-curve}.
Namely, for each $\sigma \in \Sigma$, define the \emph{conditional convexified cost curve} $c_{\sigma}(p) = \min_{\lambda} \E_{a \sim \lambda} c_{a,\sigma}$, where the minimum is over $\lambda \in \Delta_A$ such that $\E_{a \sim \lambda} p_a = p$.
\begin{lemma}
  If a plan $f:\Sigma \to A$ is feasible for Program \ref{general-dual}, then for all $\sigma$, $c_{\sigma}(p_{f(\sigma),\sigma}) = c_{f(\sigma)}$.
  Furthermore, an agent observing $\sigma$ selects a contract that is a shift of some subtangent of $c_{\sigma}$ at $p_{f(\sigma),\sigma}$.
\end{lemma}
The proof is almost identical to that of Corollary \ref{cor:contracts-elicitable} for the hidden action setting and hence omitted.
Again, the intuition is that otherwise, it is possible to more cheaply simulate the action $f(\sigma)$ via a convex combination of other actions.

Next, as we have done several times in this paper, we can ``drop'' from the solution subtangents that do not affect the objective value, without compromising feasibility.
\begin{lemma} \label{lemma:drop}
  Without loss of generality, an optimal solution to \ref{general-dual} has at most $|\Sigma|$ unique parameter pairs, in particular $\{(x_{f(\sigma),\sigma},y_{f(\sigma),\sigma}) : \sigma \in \Sigma\}$.
\end{lemma}
\begin{proof}
  In particular, the objective function and left-hand side of the incentive and participation constraints in Program \ref{general-dual} only depend on $h_{f(\sigma),\sigma}(p_{f(\sigma),\sigma})$ for each $\sigma$.
  The right-hand sides represent maxima over the subtangents at certain points, so dropping some subtangents from the solution only decrease them, maintaining feasibility.
\end{proof}

In other words, we reduce the optimality problem for plan $f$ to the following procedure:
\begin{enumerate}
  \item For each $\sigma \in \Sigma$, define the conditional convexified cost curve $c_{\sigma}(p)$.
  \item Let $V_{\sigma} = \partial c_{\sigma}(p_{f(\sigma),\sigma})$, where $\partial$ is the subdifferential (set of subgradients).
  \item Define variables $x_{\sigma} \in V_{\sigma}$, $y_{\sigma} \in \reals$ for each $\sigma$, and let $h_{\sigma}(p) = x_{\sigma} \cdot p - y_{\sigma}$.
  \item Let $G(p) = \max_{\sigma} h_{\sigma}(p)$, requiring that $G(p_{f(\sigma),\sigma}) = h_{\sigma}(p_{f(\sigma),\sigma})$ for all $\sigma$.
  \item Solve for the optimal variables subject to limited liability and participation.
\end{enumerate}

Alternatively, one can solve the following simplified linear program.
\begin{corollary} \label{cor:general-simplified}
  If plan $f$ is elicitable, then an optimal solution to the principal's minimum payment problem is given by Program \ref{general-triple}, with variables $x_{\sigma} \in \reals^{\Omega}$, $y_{\sigma} \in \reals$ for all $\sigma$ and the shorthand $h_{\sigma}(p) = x_{\sigma} \cdot p - y_{\sigma}$.
\end{corollary}
\begin{alignat}{2}
\text{}   \quad \min_{x_{\sigma}, y_{\sigma} \mid \sigma \in \Sigma} \sum_{ \sigma \in \Sigma} q(\sigma) \cdot h_{\sigma}(p_{f(\sigma),\sigma})  &           & \quad & \nonumber \\
\lplabel[general-triple]{(P6)} \text{s.t.}
\quad \sum_{\sigma \in \Sigma} q(\sigma)( h_{\sigma}(p_{f(\sigma),\sigma}) - c_{f(\sigma)}) - \kappa & \geq h_{\sigma'}(p_a) - c_{a} & \quad & \forall \sigma' \in \Sigma, \forall a \in A \label{eqn:general-triple-nosig-constr}\\
\quad h_{\sigma}(p_{f(\sigma),\sigma}) - c_{f(\sigma)} & \geq  h_{\sigma'}(p_{a,\sigma})  - c_{a}  & \quad & \forall a \in A , \forall \sigma,\sigma' \in \Sigma \label{eqn:general-triple-conditional-constr} \\
\quad \sum_{\sigma \in \Sigma} q(\sigma) ( h_{\sigma}(p_{f(\sigma),\sigma}) - c_{f(\sigma)}) - \kappa & \geq  0  & \quad & \quad \label{eqn:general-triple-partic-constr} \\
\quad h_{\sigma}(\delta_{\omega})  & \geq  0  & \quad & \forall \omega \in \Omega,\forall \sigma \in \Sigma \label{eqn:general-triple-LL-constr} \\
h_{\sigma}(p_{f(\sigma),\sigma}) & \geq  h_{\sigma'}(p_{f(\sigma),\sigma})  & & \forall \sigma,\sigma' \in \Sigma  \label{eqn:general-triple-convex-constr}
\end{alignat}

\section{Future Work}
There are many open directions.
Even in the case of information acquisition, understanding solutions with more than one available signal is a completely open and exciting direction.
For the general problem, we believe there is likely more structure to be uncovered, such as a geometric characterization of feasible plans.

There are several interesting variants of the general Contracts with Information Acquisition problem to consider.
If outcomes $\Omega$ represent principal utility, we can consider the problem of directly optimizing for a plan that maximizes utility.
This ought to make solutions simpler, because the principal will be risk-neutral.
For example, they should not prefer plans that bias the outcome distribution toward the center of the simplex --- which are harder to incentivize, because convex functions are generally lower there.

Finally, in line with recent robustness work in contract theory~\citep{carroll2015robustness,dutting2019simple}, one can ask for simple, robust solutions to the general problem, and in particular one can ask about the effectiveness of linear contracts. \bo{cite}

\vfill
\break
\bibliographystyle{ACM-Reference-Format}
\bibliography{citations}


\begin{thebibliography}{22}


\ifx \showCODEN    \undefined \def \showCODEN     #1{\unskip}     \fi
\ifx \showDOI      \undefined \def \showDOI       #1{#1}\fi
\ifx \showISBNx    \undefined \def \showISBNx     #1{\unskip}     \fi
\ifx \showISBNxiii \undefined \def \showISBNxiii  #1{\unskip}     \fi
\ifx \showISSN     \undefined \def \showISSN      #1{\unskip}     \fi
\ifx \showLCCN     \undefined \def \showLCCN      #1{\unskip}     \fi
\ifx \shownote     \undefined \def \shownote      #1{#1}          \fi
\ifx \showarticletitle \undefined \def \showarticletitle #1{#1}   \fi
\ifx \showURL      \undefined \def \showURL       {\relax}        \fi
\providecommand\bibfield[2]{#2}
\providecommand\bibinfo[2]{#2}
\providecommand\natexlab[1]{#1}
\providecommand\showeprint[2][]{arXiv:#2}

\bibitem[\protect\citeauthoryear{Alon, D\"{u}tting, and Talgam-Cohen}{Alon
  et~al\mbox{.}}{2021}]%
        {alon2021contracts}
\bibfield{author}{\bibinfo{person}{Tal Alon}, \bibinfo{person}{Paul
  D\"{u}tting}, {and} \bibinfo{person}{Inbal Talgam-Cohen}.}
  \bibinfo{year}{2021}\natexlab{}.
\newblock \bibinfo{booktitle}{\emph{Contracts with Private Cost per
  Unit-of-Effort}}.
\newblock \bibinfo{publisher}{Association for Computing Machinery},
  \bibinfo{address}{New York, NY, USA}, \bibinfo{pages}{52–69}.
\newblock
\showISBNx{9781450385541}
\urldef\tempurl%
\url{https://doi.org/10.1145/3465456.3467651}
\showURL{%
\tempurl}


\bibitem[\protect\citeauthoryear{Azar and Micali}{Azar and Micali}{2018}]%
        {azar2018computational}
\bibfield{author}{\bibinfo{person}{Pablo~D. Azar} {and} \bibinfo{person}{Silvio
  Micali}.} \bibinfo{year}{2018}\natexlab{}.
\newblock \showarticletitle{Computational principal-agent problems}.
\newblock \bibinfo{journal}{\emph{Theoretical Economics}}  \bibinfo{volume}{13}
  (\bibinfo{year}{2018}), \bibinfo{pages}{553--578}.
\newblock


\bibitem[\protect\citeauthoryear{Babaioff, Blumrosen, Lambert, and
  Reingold}{Babaioff et~al\mbox{.}}{2011}]%
        {babaioff2011only}
\bibfield{author}{\bibinfo{person}{Moshe Babaioff}, \bibinfo{person}{Liad
  Blumrosen}, \bibinfo{person}{Nicholas~S. Lambert}, {and}
  \bibinfo{person}{Omer Reingold}.} \bibinfo{year}{2011}\natexlab{}.
\newblock \showarticletitle{Only Valuable Experts Can Be Valued}. In
  \bibinfo{booktitle}{\emph{Proceedings of the 2011 ACM Conference on Economics
  and Computation}} \emph{(\bibinfo{series}{EC 2011})}.
\newblock


\bibitem[\protect\citeauthoryear{Bacon, Chen, Kash, Parkes, Rao, and
  Sridharan}{Bacon et~al\mbox{.}}{2012}]%
        {bacon2012predicting}
\bibfield{author}{\bibinfo{person}{David~F Bacon}, \bibinfo{person}{Yiling
  Chen}, \bibinfo{person}{Ian Kash}, \bibinfo{person}{David~C Parkes},
  \bibinfo{person}{Malvika Rao}, {and} \bibinfo{person}{Manu Sridharan}.}
  \bibinfo{year}{2012}\natexlab{}.
\newblock \showarticletitle{Predicting your own effort}. In
  \bibinfo{booktitle}{\emph{Proceedings of the 11th International Conference on
  Autonomous Agents and Multiagent Systems-Volume 2 (AAMAS)}}.
  \bibinfo{pages}{695--702}.
\newblock


\bibitem[\protect\citeauthoryear{Bertsimas and Tsitsiklis}{Bertsimas and
  Tsitsiklis}{1997}]%
        {bertsimas1997introduction}
\bibfield{author}{\bibinfo{person}{Dimitris Bertsimas} {and}
  \bibinfo{person}{John~N Tsitsiklis}.} \bibinfo{year}{1997}\natexlab{}.
\newblock \bibinfo{booktitle}{\emph{Introduction to Linear Optimization}}.
  Vol.~\bibinfo{volume}{1}.
\newblock \bibinfo{publisher}{Athena Scientific}.
\newblock


\bibitem[\protect\citeauthoryear{Bolton and Dewatripont}{Bolton and
  Dewatripont}{2004}]%
        {bolton2004contract}
\bibfield{author}{\bibinfo{person}{Patrick Bolton} {and}
  \bibinfo{person}{Maithas Dewatripont}.} \bibinfo{year}{2004}\natexlab{}.
\newblock \bibinfo{booktitle}{\emph{Contract Theory}}.
  Vol.~\bibinfo{volume}{1}.
\newblock \bibinfo{publisher}{The MIT Press}.
\newblock


\bibitem[\protect\citeauthoryear{Boutilier}{Boutilier}{2012}]%
        {boutilier2012eliciting}
\bibfield{author}{\bibinfo{person}{Craig Boutilier}.}
  \bibinfo{year}{2012}\natexlab{}.
\newblock \showarticletitle{Eliciting forecasts from self-interested experts:
  scoring rules for decision makers}. In \bibinfo{booktitle}{\emph{Proceedings
  of the 11th International Conference on Autonomous Agents and Multiagent
  Systems}} \emph{(\bibinfo{series}{AAMAS 2012})}. International Foundation for
  Autonomous Agents and Multiagent Systems, \bibinfo{pages}{737--744}.
\newblock


\bibitem[\protect\citeauthoryear{Carroll}{Carroll}{2015}]%
        {carroll2015robustness}
\bibfield{author}{\bibinfo{person}{Gabriel Carroll}.}
  \bibinfo{year}{2015}\natexlab{}.
\newblock \showarticletitle{Robustness and linear contracts}.
\newblock \bibinfo{journal}{\emph{American Economic Review}}
  \bibinfo{volume}{105}, \bibinfo{number}{2} (\bibinfo{year}{2015}),
  \bibinfo{pages}{536--63}.
\newblock


\bibitem[\protect\citeauthoryear{Chen and Yu}{Chen and Yu}{2021}]%
        {chen2021optimal}
\bibfield{author}{\bibinfo{person}{Yiling Chen} {and} \bibinfo{person}{Fang-Yi
  Yu}.} \bibinfo{year}{2021}\natexlab{}.
\newblock \bibinfo{title}{Optimal Scoring Rule Design}.
\newblock
\newblock
\showeprint[arxiv]{2107.07420}~[cs.GT]


\bibitem[\protect\citeauthoryear{D\"{u}tting, Roughgarden, and
  Cohen}{D\"{u}tting et~al\mbox{.}}{2020}]%
        {duetting2020complexity}
\bibfield{author}{\bibinfo{person}{Paul D\"{u}tting}, \bibinfo{person}{Tim
  Roughgarden}, {and} \bibinfo{person}{Inbal-Talgam Cohen}.}
  \bibinfo{year}{2020}\natexlab{}.
\newblock \showarticletitle{The Complexity of Contracts}. In
  \bibinfo{booktitle}{\emph{Proceedings of the Thirty-First Annual ACM-SIAM
  Symposium on Discrete Algorithms}} (Salt Lake City, Utah)
  \emph{(\bibinfo{series}{SODA '20})}. \bibinfo{publisher}{Society for
  Industrial and Applied Mathematics}, \bibinfo{address}{USA},
  \bibinfo{pages}{2688–2707}.
\newblock


\bibitem[\protect\citeauthoryear{D{\"u}tting, Roughgarden, and
  Talgam-Cohen}{D{\"u}tting et~al\mbox{.}}{2019}]%
        {dutting2019simple}
\bibfield{author}{\bibinfo{person}{Paul D{\"u}tting}, \bibinfo{person}{Tim
  Roughgarden}, {and} \bibinfo{person}{Inbal Talgam-Cohen}.}
  \bibinfo{year}{2019}\natexlab{}.
\newblock \showarticletitle{Simple versus optimal contracts}. In
  \bibinfo{booktitle}{\emph{Proceedings of the 2019 ACM Conference on Economics
  and Computation}} \emph{(\bibinfo{series}{EC 2019})}.
  \bibinfo{pages}{369--387}.
\newblock


\bibitem[\protect\citeauthoryear{Frongillo and Kash}{Frongillo and
  Kash}{2021}]%
        {frongillo2021general}
\bibfield{author}{\bibinfo{person}{Rafael~M. Frongillo} {and}
  \bibinfo{person}{Ian~A. Kash}.} \bibinfo{year}{2021}\natexlab{}.
\newblock \showarticletitle{General truthfulness characterizations via convex
  analysis}.
\newblock \bibinfo{journal}{\emph{Games and Economic Behavior}}
  \bibinfo{volume}{130} (\bibinfo{year}{2021}), \bibinfo{pages}{636--662}.
\newblock
\showISSN{0899-8256}
\urldef\tempurl%
\url{https://doi.org/10.1016/j.geb.2021.09.010}
\showDOI{\tempurl}


\bibitem[\protect\citeauthoryear{Gneiting and Raftery}{Gneiting and
  Raftery}{2007}]%
        {gneiting2007strictly}
\bibfield{author}{\bibinfo{person}{Tilman Gneiting} {and}
  \bibinfo{person}{Adrian~E. Raftery}.} \bibinfo{year}{2007}\natexlab{}.
\newblock \showarticletitle{Strictly proper scoring rules, prediction, and
  estimation}.
\newblock \bibinfo{journal}{\emph{J. Amer. Statist. Assoc.}}
  \bibinfo{volume}{102}, \bibinfo{number}{477} (\bibinfo{year}{2007}),
  \bibinfo{pages}{359--378}.
\newblock


\bibitem[\protect\citeauthoryear{Grossman and Hart}{Grossman and Hart}{1983}]%
        {grossman1983analysis}
\bibfield{author}{\bibinfo{person}{Sanford~J. Grossman} {and}
  \bibinfo{person}{Oliver~D. Hart}.} \bibinfo{year}{1983}\natexlab{}.
\newblock \showarticletitle{An Analysis of the Principal-Agent Problem}.
\newblock \bibinfo{journal}{\emph{Econometrica}} \bibinfo{volume}{51},
  \bibinfo{number}{1} (\bibinfo{year}{1983}), \bibinfo{pages}{7--45}.
\newblock
\showISSN{00129682, 14680262}
\urldef\tempurl%
\url{http://www.jstor.org/stable/1912246}
\showURL{%
\tempurl}


\bibitem[\protect\citeauthoryear{Guruganesh, Schneider, and Wang}{Guruganesh
  et~al\mbox{.}}{2021}]%
        {guruganesh2021contracts}
\bibfield{author}{\bibinfo{person}{Guru Guruganesh}, \bibinfo{person}{Jon
  Schneider}, {and} \bibinfo{person}{Joshua~R. Wang}.}
  \bibinfo{year}{2021}\natexlab{}.
\newblock \bibinfo{booktitle}{\emph{Contracts under Moral Hazard and Adverse
  Selection}}.
\newblock \bibinfo{publisher}{Association for Computing Machinery},
  \bibinfo{address}{New York, NY, USA}, \bibinfo{pages}{563–582}.
\newblock
\showISBNx{9781450385541}


\bibitem[\protect\citeauthoryear{Hartline, Li, Shan, and Wu}{Hartline
  et~al\mbox{.}}{2020}]%
        {li2020optimization}
\bibfield{author}{\bibinfo{person}{Jason~D. Hartline}, \bibinfo{person}{Yingkai
  Li}, \bibinfo{person}{Liren Shan}, {and} \bibinfo{person}{Yifan Wu}.}
  \bibinfo{year}{2020}\natexlab{}.
\newblock \bibinfo{title}{Optimization of Scoring Rules}.
\newblock
\newblock
\showeprint[arxiv]{2007.02905}~[cs.GT]


\bibitem[\protect\citeauthoryear{Hiriart-Urrut and
  Lemar\'{e}chal}{Hiriart-Urrut and Lemar\'{e}chal}{2001}]%
        {hiriarturrut2001fundamentals}
\bibfield{author}{\bibinfo{person}{Jean-Baptiste Hiriart-Urrut} {and}
  \bibinfo{person}{Claude Lemar\'{e}chal}.} \bibinfo{year}{2001}\natexlab{}.
\newblock \bibinfo{booktitle}{\emph{Fundamentals of Convex Analysis}}.
\newblock \bibinfo{publisher}{Springer}.
\newblock


\bibitem[\protect\citeauthoryear{Mas-Colell, Whinston, and Green}{Mas-Colell
  et~al\mbox{.}}{1995}]%
        {mas1995microeconomic}
\bibfield{author}{\bibinfo{person}{Andreu Mas-Colell},
  \bibinfo{person}{Michael~Dennis Whinston}, {and} \bibinfo{person}{Jerry~R
  Green}.} \bibinfo{year}{1995}\natexlab{}.
\newblock \bibinfo{booktitle}{\emph{Microeconomic theory}}.
  Vol.~\bibinfo{volume}{1}.
\newblock \bibinfo{publisher}{Oxford University Press}.
\newblock


\bibitem[\protect\citeauthoryear{Oesterheld and Conitzer}{Oesterheld and
  Conitzer}{2020a}]%
        {oesterheld2020decision}
\bibfield{author}{\bibinfo{person}{Caspar Oesterheld} {and}
  \bibinfo{person}{Vincent Conitzer}.} \bibinfo{year}{2020}\natexlab{a}.
\newblock \showarticletitle{Decision Scoring Rules}. In
  \bibinfo{booktitle}{\emph{Web and Internet Economics}}.
  \bibinfo{publisher}{Springer International Publishing},
  \bibinfo{address}{Cham}, \bibinfo{pages}{468}.
\newblock
\showISBNx{978-3-030-64946-3}


\bibitem[\protect\citeauthoryear{Oesterheld and Conitzer}{Oesterheld and
  Conitzer}{2020b}]%
        {oesterheld2020minimum}
\bibfield{author}{\bibinfo{person}{Caspar Oesterheld} {and}
  \bibinfo{person}{Vincent Conitzer}.} \bibinfo{year}{2020}\natexlab{b}.
\newblock \showarticletitle{Minimum-Regret Contracts for Principal-Expert
  Problems}. In \bibinfo{booktitle}{\emph{Web and Internet Economics}}
  \emph{(\bibinfo{series}{WINE 2020})}.
\newblock


\bibitem[\protect\citeauthoryear{Savage}{Savage}{1971}]%
        {savage1971elicitation}
\bibfield{author}{\bibinfo{person}{Leonard~J. Savage}.}
  \bibinfo{year}{1971}\natexlab{}.
\newblock \showarticletitle{Elicitation of personal probabilities and
  expectations}.
\newblock \bibinfo{journal}{\emph{J. Amer. Statist. Assoc.}}
  \bibinfo{volume}{66}, \bibinfo{number}{336} (\bibinfo{year}{1971}),
  \bibinfo{pages}{783--801}.
\newblock


\bibitem[\protect\citeauthoryear{Tadelis and Segal}{Tadelis and Segal}{2005}]%
        {tadelis2005lectures}
\bibfield{author}{\bibinfo{person}{Steve Tadelis} {and} \bibinfo{person}{Ilya
  Segal}.} \bibinfo{year}{2005}\natexlab{}.
\newblock \bibinfo{booktitle}{\emph{Lectures in Contract Theory}}.
  Vol.~\bibinfo{volume}{1}.
\newblock
\urldef\tempurl%
\url{http://faculty.haas.berkeley.edu/stadelis/Econ_206_notes_2006.pdf}
\showURL{%
\tempurl}


\end{thebibliography}

\appendix

\section{Information Acquisition}
\subsection{Omitted proofs}
At several points, we used the following technical lemma.
Given any point $p$, we show that $p_0$ is in the convex hull of $p$ and at most $|\Omega|-1$ corners of the simplex.
In constructing the optimal solution $\bar{G}^*$ to Program \ref{IA-dual}, this  will allow us to pick out an $h_{\omega}^*$ that achieves the maximum at $p$ using the missing corner.
Lemma \ref{lemma:p-p0-interior} is illustrated in Figure \ref{fig:points-conv-hull}.

\begin{figure}[H]
\centering
\caption{An illustration of Lemma \ref{lemma:p-p0-interior}: given any point $p$ in the simplex, we can represented $p_0$ as a convex combination of $p$ (with positive weight) and all but one corners of the simplex.}
\label{fig:points-conv-hull}
\begin{tikzpicture}
  \fill (7.5,6.2) circle (2pt);
  \node [draw] at (7.5,5.8){\tiny $p_0$};
  \fill (8,7.4) circle (2pt);
  \node [draw] at (8.3,7.4){\tiny $p$};
  \draw[thick,dotted] (8,7.4) -- (5,5);
    \draw[thick, dotted] (8,7.4) -- (10,5);
  \draw (5,5) node[anchor=north]{$\delta_x$}
  -- (10,5) node[anchor=north]{$\delta_y$}
  -- (7.5,10) node[anchor=south]{$\delta_z$}
  -- cycle;
\end{tikzpicture}
\end{figure}
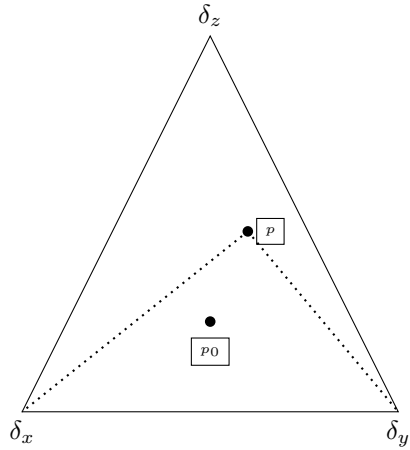

\begin{lemma}\label{lemma:p-p0-interior}
  For any $p \in \Delta_{\Omega}$, there exists $\omega \in \Omega$ s.t. $p_0\in \text{convhull}(\{\delta_{\omega'} : \omega' \neq \omega\} \bigcup \{p\} )$.
  In particular, there exist nonnegative numbers $\beta, \{\beta_{\omega'} : \omega' \in \Omega, \omega' \neq \omega\}$ summing to one with $\beta > 0$ such that $p_0 = \beta p + \sum_{\omega' \neq \omega} \beta_{\omega'} \delta_{\omega'}$.
\end{lemma}
\begin{proof}
  Fix $p$ and let $\omega= \arg\max_{\omega' \in \Omega} \frac{p(\omega')}{p_0(\omega')}$, breaking ties arbitrarily.
  (Recall by assumption that $p_0$ has full support, so $\omega$ is well-defined.)
  Let $\beta = \frac{p_0(\omega)}{p(\omega)}$, observing that $0 < \beta \leq 1$.
  For $\omega' \neq \omega$, let $\beta_{\omega'} = p_0(\omega') - \beta p(\omega')$.
  
  First, we show the numbers form a probability distribution.
  We already have $\beta > 0$, and for $\omega' \neq \omega$, we have $\beta_{\omega'} = p_0(\omega') \left(1 - \beta \frac{p(\omega')}{p_0(\omega')}\right) \geq 0$ by definition of $\beta$.
  Meanwhile,
  \begin{align*}
    \beta + \sum_{\omega' \neq \omega} \beta_{\omega'}
    &= \beta + \sum_{\omega' \neq \omega} \left(p_0(\omega') - \beta p(\omega')\right)  \\
    &= \beta + (1 - p_0(\omega)) - \beta(1 - p(\omega))  \\
    &= 1 - p_0(\omega) + \beta p(\omega)  \\
    &= 1 .
  \end{align*}
  Now, all that remains is to show is that $p_0$ is indeed equal to the quantity $r := \beta p + \sum_{\omega' \neq \omega} \beta_{\omega'} \delta_{\omega'}$.
  We have $r(\omega) = \beta p(\omega) = p_0(\omega)$, as desired.
  And for $\omega' \neq \omega$, we have $r(\omega') = \beta p(\omega') + \beta_{\omega'} = p_0(\omega')$.
\end{proof}

\subsection{Clarification for Remark \ref{remark:how-we-IA}} \label{app:IA-how-opt}

In this subsection we illustrate how we arrived at the closed form solution for Program \ref{IA-dual} in the Information Acquisition(IA) model in \ref{sec:info_acq}. The proofs of these claims are omitted as they are not being used in the main text.

Figure \ref{fig:cone} illustrates why pointed polyhedral cones are optimal.
Figure \ref{fig:min-payment} shows that the limited liability constraint is binding at all corners in an optimal solution.
Figure \ref{fig:prior-value} illustrates why an optimal solution has $G(p_0) = 1$.
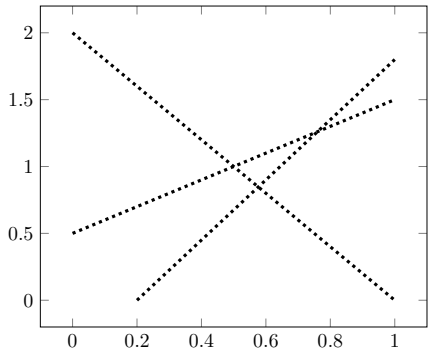
\begin{figure}[H]
\centering
\caption{Claim: There is always an optimal solution to Program \ref{IA-dual} that is a pointed polyhedral cone (PPC).}
\label{fig:cone}
\begin{subfigure}{.5\textwidth}
  \centering
  \captionsetup{width=.9\linewidth,font=small}
\begin{tikzpicture} [scale=0.75]
\begin{axis}
    \addplot [mark=none,  black,   ultra thick, dotted] coordinates { (0,2) (1, 0)};
    \addplot [mark=none,  black,   ultra thick, dotted] coordinates { (0,0.5) (1,1.5)};
    \addplot [mark=none,  black,   ultra thick, dotted] coordinates { (0.2,0) (1,1.8)};
\end{axis}
\end{tikzpicture}
\caption{Here we start with $G$, the pointwise maximum over the three pictured affine functions (black dotted lines).}
\end{subfigure}%
\begin{subfigure}{.5\textwidth}
  \centering
  \captionsetup{width=.9\linewidth,font=small}
\begin{tikzpicture} [scale=0.75]
\begin{axis}
    \addplot [mark=none,  red,   ultra thick] coordinates { (0,2) (1, 0)};
    \addplot [mark=none,  red,   ultra thick] coordinates { (0,0.5) (1,1.5)};
    \addplot [mark=none,  red,   ultra thick] coordinates { (0.1,0.1) (1,2.1)};
    \addplot [mark=none,  black,   ultra thick, dotted] coordinates { (0,2) (1, 0)};
    \addplot [mark=none,  black,   ultra thick, dotted] coordinates { (0,0.5) (1,1.5)};
    \addplot [mark=none,  black,   ultra thick, dotted] coordinates { (0.2,0) (1,1.8)};
\end{axis} 
\end{tikzpicture}
\caption{We modify $G$ by shifting its supporting hyperplanes up up to pass through the point $(p_0,1)$, resulting in the red solid lines. The pointwise maximum over these functions is a cone, as all of them pass through the point $(p_0,1)$. If the original $G$ was feasible, then the resulting $G$ is feasible and its expected value (objective value in Program \ref{IA-dual}) only increases.}
\end{subfigure}%
\end{figure}

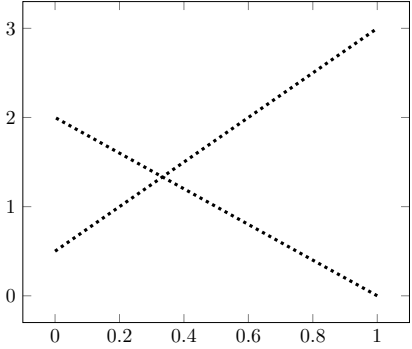
\begin{figure}[H]
\centering
\caption{Claim: At an optimal solution, for all $\omega$, the limited liability constraint is binding.}
\label{fig:min-payment}
\begin{subfigure}{.5\textwidth}
  \centering
  \captionsetup{width=.9\linewidth , font=small}
\begin{tikzpicture} [scale=0.75]
\begin{axis}
    \addplot [mark=none,  black,   ultra thick, dotted] coordinates { (0,0.5) (1,3)};
    \addplot [mark=none,  black,   ultra thick, dotted] coordinates { (1,0) (0,2)};
\end{axis} 
\end{tikzpicture}
\caption{Payment structure, $G$,  having an outcome (left) with minimum payment strictly positive.}
\end{subfigure}%
\begin{subfigure}{.5\textwidth}
  \centering
  \captionsetup{width=.9\linewidth ,font=small}
\begin{tikzpicture} [scale=0.75]
\begin{axis}
    \addplot [mark=none,  red,   ultra thick] coordinates { (1,0) (0,2)};
    \addplot [mark=none,  blue,  ultra thick] coordinates { (0,0) (1,2.5)};
    \addplot [mark=none,  red,  ultra thick] coordinates { (0,0) (1,4)};
    \addplot [mark=none,  black,   ultra thick, dotted] coordinates { (0,0.5) (1,3)};
    \addplot [mark=none,  black,   ultra thick, dotted] coordinates { (1,0) (0,2)};
\end{axis} 
\end{tikzpicture}
\caption{First, shift the offending hyperplanes of $G$ down to make minimum payment equal $0$, resulting in the blue line. Then, scale them up to retain the original expected payment for reporting prior, resulting in the red line. None of the constraints are violated, and the expected value of $G$ increases.}
\end{subfigure}%
\end{figure}

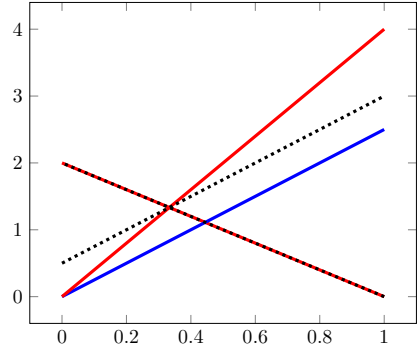
\begin{figure}[H]
\centering
\caption{Claim: at an optimal solution, the constraint $G(p_0) \leq 1$ is binding.}
\label{fig:prior-value}
\begin{subfigure}{.5\textwidth}
  \centering
  \captionsetup{width=.9\linewidth,font=small}
\begin{tikzpicture} [scale=0.75]
\begin{axis}
    \addplot [mark=none,  black,   ultra thick, dotted] coordinates { (0,0) (1,1)};
    \addplot [mark=none,  black,   ultra thick, dotted] coordinates { (1,0) (0,1)};
\end{axis} 
\end{tikzpicture}
\caption{Here $G(p_0) < 1$, i.e. the supporting hyperplanes lie strictly below the point $(p_0,1)$.}
\end{subfigure}%
\begin{subfigure}{.5\textwidth}
  \centering
  \captionsetup{width=.9\linewidth,font=small}
\begin{tikzpicture} [scale=0.75]
\begin{axis}
    \addplot [mark=none,  black,   ultra thick, dotted] coordinates { (0,0) (1,1)};
    \addplot [mark=none,  black,   ultra thick, dotted] coordinates { (1,0) (0,1)};
    \addplot [mark=none,  red,  ultra thick] coordinates { (0,0) (1,2)};
    \addplot [mark=none,  red,  ultra thick] coordinates { (1,0) (0,2)};
\end{axis} 
\end{tikzpicture}
\caption{Scale up all supporting hyperplanes, i.e. scale up $G$, to obtain an expected payment of $1$ for reporting the prior. None of the constraints are violated, but the expected value of payments increase.}
\end{subfigure}%
\end{figure}

\section{Contracts} \label{app:contracts}

\subsection{Computing all optimal contracts}

The next result shows that $a^*$ is elicited by a menu $G$ if and only if $G$ has the following form: take a menu $G'$ that lies everywhere below the convexified cost curve $c$, and shift it upward by some amount $\beta$ to obtain $G$.
Furthermore, $G'$ must equal $c$ at $p_{a^*}$, where both must equal the cost $c_{a^*}$.
Examples appear in Figure \ref{fig:strict-weak-elicit}.
\begin{prop} \label{prop:contracts-elicitable}
  Let a contracts setting be given and let $c$ be the convexified cost curve.
  A subdifferentiable convex menu $G$ elicits an action $a^*$ if and only if there exists $\beta \geq 0$ such that: \emph{(a)} $G(p) - \beta \leq c(p)$ pointwise on $\P$, and \emph{(b)} $G(p_{a^*}) - \beta = c(p_{a^*}) = c_{a^*}$.
\end{prop}
\begin{proof}
  ($\implies$)
  Suppose that the subdifferentiable convex $G$ elicits $a^*$.
  Let $\beta = G(p_{a^*}) - c_{a^*}$.
  By the participation constraint (\ref{eqn:contracts-partic-constr}), $\beta \geq 0$.
  By the incentive constraint (\ref{eqn:contracts-incent-constr}), $G'(p_a) \leq c_a$ for all $a \in A$.
  So $G'$ is a convex function that lies below all points in $X = \{(p_a, c_a) : a \in A\}$.
  By definition, $c$ is the pointwise maximum over all such convex functions, so $G'(p) \leq c(p)$ for all $p$, proving \emph{(a)}.
  Furthermore, we now have $c_{a^*} = G'(p_{a^*}) \leq c(p_{a^*}) \leq c_{a^*}$, so these are all equalities, proving \emph{(b)}.
  
  ($\impliedby$)
  Let $\beta \geq 0$ and let $G'(p) = G(p) - \beta$; suppose \emph{(a)} and \emph{(b)} hold.
  Then \emph{(b)} immediately implies the participation constraint (\ref{eqn:contracts-partic-constr}).
  It also implies $G(p_{a^*}) - c_{a^*} = \beta$.
  Meanwhile, \emph{(a)} and the definition of $c$ imply $G(p_a) - \beta \leq c(p_a) \leq c_a$.
  Rearranging, $G(p_a) - c_a \leq \beta$, which is the incentive constraint (\ref{eqn:contracts-incent-constr}).
\end{proof}

\begin{corollary} \label{cor:contracts-elicitable}
  An action $a^*$ is elicitable if and only if $c(p_{a^*}) = c_{a^*}$, i.e. if the point $(p_{a^*},c_{a^*})$ lies on the lower boundary of the convex hull of $\{(p_a,c_a) : a \in A\}$.
\end{corollary}
\begin{proof}
  Condition \emph{(b)} in Proposition \ref{prop:contracts-elicitable} holds if and only if $c(p_{a^*}) = c_{a^*}$, in which case we can take $G = c$ to satisfy condition \emph{(a)} as well.
\end{proof}

If $a^*$ is elicitable, then even when we add the additional requirement of limited liability, there is always a feasible way to incentivize it.
In particular, if we have a $G'$ that satisfies conditions \emph{(a), (b)} in Proposition \ref{prop:contracts-elicitable}, then there exists a $\beta \geq 0$ such that $G' + \beta$ satisfies the limited liability constraint (\ref{eqn:contracts-LL-constr}) as well.
\begin{corollary} \label{cor:contracts-feasible}
  In the contracts setting, given a desired  action $a^*$, there exists a feasible solution $G$ to Program \ref{contracts-main} if and only if  the convexified cost curve satisfies $c(p_{a^*}) = c_{a^*}$.
  In this case, for any subdifferentiable convex $G$, there exists $\beta' \geq 0$ such that $G + \beta'$ is feasible for \ref{contracts-main} if and only if there exists $\beta \geq 0$ such that conditions \emph{(a), (b)} of Proposition \ref{prop:contracts-elicitable} are satisfied.
\end{corollary}
\begin{proof}
  We have from Proposition \ref{prop:contracts-elicitable} that $G$ satisfies constraints (\ref{eqn:contracts-incent-constr}) and (\ref{eqn:contracts-partic-constr}) if and only if it is any upward shift of some such $G'$.
  Next, we observe that because the set of subtangents is bounded, there exists a sufficiently large upward shift $\beta$ such that the final constraint, limited liability (\ref{eqn:contracts-LL-constr}), is satisfied as well.\bo{wouldn't hurt to have more detail, but this is ok}
\end{proof}

Now that we have a geometric characterization of feasibility, we turn to optimality.
Here the key point, as discussed above, is that there are optimal menus consisting of just one contract.
That contract must be a shift of a subtangent of $c$ at $p_{a^*}$, and can be found by optimizing over the set of subgradients of $c$ at $p_{a^*}$.
\begin{prop} \label{prop:contracts-minimal-opt}
  If $a^*$ is elicitable, then Algorithm \ref{alg:contracts-minimal} computes an optimal solution to Program \ref{contracts-main} consisting of a single contract.
\end{prop}



\begin{proof}
  \bo{Am sure we can improve this presentation}
  Let $G$ be any optimal menu.
  By Corollary \ref{cor:contracts-feasible}, $G$ is feasible if and only if, for some $\beta \geq 0$, letting $G' = G - \beta$, we have $G'(p) \leq c(p)$ pointwise with $G'(p_{a^*}) = c(p_{a^*}) = c_{a^*}$.
  In particular, $G'$ has some subtangent contract $t'$ at $p_{a^*}$, and by definition, $t'$ must be a subtangent of $c$ at $p_{a^*}$ as well, i.e. for some $v \in \partial c(p_{a^*})$,
  \[ t'(\omega) = c(p_{a^*}) + v \cdot (\delta_{\omega} - p_{a^*}) . \]
  So $G$ has a subtangent $t = t' + \beta$.
  
  Observe that the objective value for $G$ is $G(p_{a^*}) = \bar{t}(p_{a^*})$.
  Also observe that if we drop all other contracts from the menu besides $t$, all constraints continue to hold, and the objective value is unchanged.
  So there are optimal menus of the form $T = \{t\}$ where $t$ is an upward shift of a subtangent of $c$ at $p_{a^*}$.
  
  We now argue that there is an optimal single-contract menu where the subgradient $v$ is in $V_{a^*}$ and that $\beta$ is as calculated in Algorithm \ref{alg:contracts-minimal}.
  In particular, by limited liability, $\min_{\omega} t(\omega) \geq 0$, i.e.
  \begin{align*}
    -\beta &\leq c(p_{a^*})  + \min_{\omega} v(\omega) - v\cdot p_{a^*} .
  \end{align*}
  \bo{this is where it could use more clarity}
  By definition, if $v \in V_{a^*}$, the shift $\beta$ is minimized over all $\partial c(p_{a^*})$, so in particular the above inequality holds.
  And the other constraints continue to hold as well, so the resulting contract is also optimal.
\end{proof}

Additionally, one can generally add a number of additional contracts to the menu without compromising optimality.
We can view this process as taking the optimal subtangent, an affine $G$, and ``convexifying'' it further by adding other contracts that do not violate limited liability and do not lie above $c$ (after shifting down by the appropriate $\beta$).
The result is Algorithm \ref{alg:contracts-full}.
We note that Algorithm 2 can in general add contracts to the menu that are strictly dominated (i.e. never chosen by the agent), so one can optionally post-process the output to remove any dominated contracts.

\begin{prop} \label{prop:contracts-all-opt}
  If $a^*$ is elicitable, then every optimal solution to Program \ref{contracts-main} is computed by Algorithm \ref{alg:contracts-full}, for some choices of tiebreakers.
\end{prop}

\begin{algorithm}
  \caption{Computing any optimal contracts menu $T$ and its convex representation $G$.}
  \label{alg:contracts-full}
  Given a contracts problem, let $c$ be the convexified cost curve and $a^*$ the desired elicitable action. \\
  Define $V_{a^*} = \arg\max_{v \in \partial c(p_{a^*})} \min_{\omega}   v(\omega) - v \cdot p_{a^*} .$ \tcp*{the optimal subgradients of $c$ at $p_{a^*}$}
  Let $v \in V_{a^*}$. \\
  Define $\beta = -\left( c(p_{a^*}) + \min_{\omega} v(\omega) - v \cdot p_{a^*} \right)$. \tcp*{the shift required for limited liability}
  If $\beta < 0$: \tcp*{limited liability does not bind}
    \quad Let $\beta = 0$. \\
    \quad Let $v \leftarrow $ Algorithm \ref{alg:contracts-partic-subtans}$(c,p_{a^*})$. \tcp*{subgradients outside $V_{a^*}$ are also optimal}
  Let $t(\omega) = c(p_{a^*}) + v \cdot (\delta_{\omega} - p_{a^*}) + \beta$. \\
  Let $T = \{t\} \cup $ Algorithm \ref{alg:contracts-extras}$(c,\beta)$. \\
  Define $G(p) = \max_{t \in T} \bar{t}(p)$.
\end{algorithm}

\begin{algorithm}
  \caption{On input $c,p_{a^*}$, return subtangents not violating LL.}
  \label{alg:contracts-partic-subtans}
  Define $f(v) = c(p_{a^*}) + \min_{\omega} v(\omega) - v \cdot p_{a^*}$.  \tcp*{minimum payment of subtangent with slope $v$}
  Define $U_{a^*} = \{v \in \partial c(p_{a^*}) ~:~ f(v) \geq 0 \}$.  \tcp*{subgradients whose contracts don't violate LL}
  Let $v \in U_{a^*}$. \\
  Return $v$.
\end{algorithm}

\begin{algorithm}
  \caption{On input $c,\beta$, return subset of contracts satisfying LL and whose downward-$\beta$ shifts lie below $c$.}
  \label{alg:contracts-extras}
  Let $A = \{t ~:~ \bar{t}(p) \leq c(p) ~ (\forall p)\}$.  \tcp*{contracts lying below $c$}
  Let $A' = \{t + \beta ~:~ t \in A\}$. \\
  Let $B = \{t \in A' ~:~ \min_{\omega} t(\omega) \geq 0\}$. \tcp*{shifted contracts satisfying limited liability}
  Return an arbitrary subset of $B$.
\end{algorithm}
First, we give a lemma.
\begin{lemma} \label{lemma:contracts-alg-subset}
  Let $T = \{t\}$ be an optimal menu for Program \ref{contracts-main}, where $t(\omega) = c(p_{a^*}) + v \cdot (\delta_{\omega} - p_{a^*}) + \beta$ for some $\beta \geq 0$.
  Then a menu $T'$ containing $t$ is optimal if and only if $T' = T \cup $ Algorithm \ref{alg:contracts-extras}$(c,\beta)$ for some choice of tiebreaker.
\end{lemma}
\begin{proof}
  First, we show that $T \cup $Algorithm \ref{alg:contracts-extras}$(c,\beta)$ is always optimal.
  
  Let $G(p) = \bar{t}(p)$ be the convex representation, let $T' = T \cup $ Algorithm \ref{alg:contracts-extras}$(c,\beta)$, and let $G'(p) = \max_{t' \in T'} \bar{t'}(p)$ be its convex representation.
  
  We have $G(p_{a^*}) = \bar{t}(p_{a^*}) = c(p_{a^*}) + \beta$.
  Further, by definition of $A,A',B$ in Algorithm \ref{alg:contracts-extras}, for all $t' \in T'$, we have $\bar{t'}(p_{a^*}) \leq c(p_{a^*}) + \beta$.
  So the objective values of $G'$ and $G$ are the same.
  Now we just need to show that $G'$ is feasible.
  By construction of $B$, all elements of $T'$ satisfy the limited liability constraint.
  And it follows immediately that $G'$ satisfies the conditions of Proposition \ref{prop:contracts-elicitable}, so it elicits $p_{a^*}$, i.e. satisfies the incentive and participation constraints.
  
  Now, we show any optimal menu containing $t$ can be produced.
  
  Let $G(p) = \max_{i \in I} h_i(p)$ be an optimal solution of Program \ref{contracts-main} where, for some $j \in I$, we have $h_j(p) = \bar{t}(p)$.
  As it is feasible, $\min_\omega h_i(\omega) \geq 0 ,\forall i$.
  Because $\{t\}$ is optimal, we have $h_j(p_{a^*})\geq h_i(p_{a^*}), \forall i$, that is, $h_j$ is tangent to $G$ at $p_{a^*}$.
  
  Corollary \ref{cor:contracts-feasible} implies that there exists a $\beta'$ such that $h_i(p)\leq G(p) \leq c(p) +\beta' ,\forall i\in I, \forall p$ and $h_j(p_{a^*}) = c(p_{a^*}) +\beta'$.
  We observe that $\bar{t}(p_{a^*}) = c(p_{a^*}) + \beta$.
  Therefore, $\beta' = \beta$.
  This gives for each $i$, $h_i\in A'$ and $h_i\in B$, which implies it is one possible among $T\cup$ Algorithm \ref{alg:contracts-extras}$(c,\beta)$.
  As $G$ was any optimal solution containing $t$, we proved the lemma.
\end{proof}

\begin{proof}[Proof of Proposition \ref{prop:contracts-all-opt}]
  We show that a contract $t$ is produced in Algorithm \ref{alg:contracts-full} if and only if $\{t\}$ is an optimal solution of Program \ref{contracts-main}.
  This combined with Lemma \ref{lemma:contracts-alg-subset} proves the result.
  
  By Proposition \ref{prop:contracts-minimal-opt}, Algorithm \ref{alg:contracts-minimal} computes an optimal single-contract menu $T = \{t\}$.
  In particular, the optimal objective value is $\bar{t}(p_{a^*}) = c(p_{a^*}) + \max\{0, \beta\}$.
  There are two cases in Algorithm \ref{alg:contracts-full}.
  
  In the case that $\beta \geq 0$, we observe that $\min_{\omega} t(\omega) = 0$, i.e. limited liability is binding.
  Using Corollary \ref{cor:contracts-feasible}, a singleton optimal is of the form $G(p) = \bar{t}(p) = c(p_{a^*})+ dt\cdot (p-p_{a^*}) + \beta$, where $dt\in \partial c(p_{a^*})$ and $\beta =  - (c(p_{a^*}) + \min_\omega dt(\omega) - dt\cdot p_{a^*} )$.
  We argue in this case that any optimal menu must contain some $t'$ with a subgradient in $V_{a^*}$.
  Let us assume the contrary.
  So $dt\in \partial c(p_{a^*}) \setminus V_{a^*}$.
  By the definition of $V_{a^*}$, there exists $v \in V_{a^*}$ s.t. $\min_\omega v(\omega) - v \cdot p_{a^*} > \min_\omega dt(\omega) - dt \cdot p_{a^*}$.
  This implies
  \begin{align*}
        c(p_{a^*}) - \min_\omega (c(p_{a^*}) + v(\omega) - v\cdot p_{a^*}) \ &< c(p_{a^*}) - \min_\omega (c(p_{a^*}) + dt(\omega) - dt\cdot p_{a^*})\\
        \implies c(p_{a^*}) + \beta & < G(p_{a^*})
  \end{align*}
  which is a contradiction.
  Therefore, the optimal singleton menus are exactly of the form $T$ computed in Algorithm \ref{alg:contracts-full}.
  
  In the case that $\beta < 0$, intuitively, no shift is required to satisfy limited liability.
  In this case, the participation constraint binds, i.e. we observe that at the optimal solution $\{t\}$ of Algorithm \ref{alg:contracts-minimal}, we have $\bar{t}(p_{a^*}) = c(p_{a^*})$.
  Thus, a singleton menu $\{t'\}$ is optimal if and only if $\bar{t'}(p_{a^*}) = c(p_{a^*})$ and $t'$ satisfies limited liability.
  This holds if and only if Algorithm \ref{alg:contracts-partic-subtans} can return $\nabla t'$.
  So in the case $\beta < 0$, it follows a singleton menu $\{t\}$ is optimal if and only if it is defined in Algorithm \ref{alg:contracts-full} for some tiebreaker.
\end{proof}

\subsection{Strict elicitation} \label{app:contracts-strict}

Recall that $v$ is a \emph{strict} subgradient of $G$ at $p$ if $G(p') > G(p) + v \cdot (p' - p)$ for all $p' \neq p$.
We let $\bar{\partial} G(p)$ denote the set of strict subgradients of $G$ at $p$.

\begin{lemma} \label{lemma:contracts-strict-iff}
  Let $a^*$ be elicitable.
  Then $(p_{a^*},c_{a^*})$ is a lower vertex and $v \in \bar{\partial} c(p_{a^*})$ $\iff$ the contract $t(\omega) = c(p_{a^*}) + v \cdot (\delta_{\omega} - p_{a^*})$ satisfies all incentive constraints (\ref{eqn:contracts-incent-constr}) for $a \neq a^*$ with strict inequality.
\end{lemma}
\begin{proof}
  ($\implies$)
  Suppose $(p_{a^*},c_{a^*})$ is a lower vertex and suppose $v \in \bar{\partial} c(p_{a^*})$.
  Let $t(\omega) = c(p_{a^*}) + v \cdot (\delta_{\omega} - p_{a^*})$.
  By definition, $\bar{t}(p_a) < c(p_a) \leq c_a$ for all $a \neq a^*$, showing that the incentive constraints are satisfied strictly.
  
  ($\impliedby$)
  We have $\bar{t}(p_{a^*}) = c(p_{a^*})$ and $\bar{t}(p_a) = c(p_{a^*}) + v \cdot (p_a - p_{a^*})$.
  So the incentive constraint implies
  \[ c(p_{a^*}) + v \cdot (p_a - p_{a^*}) \leq c_a , \]
  with strict inequality unless $a=a^*$.
  Consider any $p \neq p_{a^*}$.
  We can write $p = \E_{a \sim \lambda} p_a$, where $\lambda(a^*) \neq 1$.
  Taking the expectation of both sides of the above equation, we obtain $c(p_{a^*}) + v \cdot (p - p_{a^*}) < \E_{a \sim \lambda} c_a$.
  The minimum of the right hand side over all valid $\lambda$ is $c(p)$, so $v$ is a strict subgradient.
  Now if $\lambda \in \Delta_{A \setminus \{a^*\}}$ with $\E_{a \sim \lambda} p_a = p_{a^*}$, then taking the expectation of both sides yields $c(p_{a^*}) < \E_{a \sim \lambda} c_a$.
  This implies $c(p_{a^*}) = c_a$ (otherwise, we obtain a contradiction) and that $(p_{a^*},c_{a^*})$ is a lower vertex.
\end{proof}

\begin{corollary} \label{cor:strict-char}
  An action $a^*$ is strictly elicitable if and only if $(p_{a^*}, c_{a^*})$ is a lower vertex of $\mathrm{convhull}(\{(p_a,c_a) : a \in A\})$.
\end{corollary}
\begin{proof}
  ($\implies$)
  If $a^*$ is strictly elicited by $G$, then in particular $G$ has a contract $t$ satisfying all incentive constraints for $a \neq a^*$ with strict inequality, and Lemma \ref{lemma:contracts-strict-iff} implies the claim.
  
  ($\impliedby$)
  If $(p_{a^*},c_{a^*})$ is a lower vertex, then $\bar{\partial} c(p_{a^*})$ is nonempty by \cite{bertsimas1997introduction} (Definition 2.7 and Theorem 2.3).
  Therefore, Lemma \ref{lemma:contracts-strict-iff} supplies a contract $t$ satisfying all incentive constraints strictly for $a \neq a^*$.
  By construction $\bar{t}(p_{a^*}) = c(p_{a^*})$, so participation is satisfied.
  So the menu $\{t\}$ strictly elicits $a^*$.
\end{proof}

\begin{algorithm}
  \caption{Computing a contract that strictly elicits $a^*$ almost-optimally.}
  \label{alg:contracts-strict}
  \bo{extend this to add a small quadratic as well?}
  Given a contracts problem and $\epsilon > 0$, let $c$ be the convexified cost curve and $a^*$ the desired elicitable action. \\
  Define $V_{a^*} = \arg\max_{v \in \partial c(p_{a^*})} \min_{\omega}   v(\omega) - v \cdot p_{a^*} .$ \tcp*{the optimal subgradients of $c$ at $p_{a^*}$}
  Let $v_1 \in V_{a^*}$. \\
  Define $\beta_1 = \max\left\{0 ~,~ -\left( c(p_{a^*}) + \min_{\omega} v_1(\omega) - v_1 \cdot p_{a^*} \right) \right\}$. \tcp*{the shift required for limited liability}
  Define $t_1(\omega) = c(p_{a^*}) + v_1 \cdot (\delta_{\omega} - p_{a^*}) + \beta_1$. \\
  Let $v_2 \in \bar{\partial} c(p_{a^*})$. \tcp*{a strict subgradient}
  Define $\beta_2 = \max\left\{0 ~,~ -\left( c(p_{a^*}) + \min_{\omega} v_2(\omega) - v_2 \cdot p_{a^*} \right) \right\}$. \tcp*{the shift required for limited liability}
  Define $t_2(\omega) = c(p_{a^*}) + v_2 \cdot (\delta_{\omega} - p_{a^*}) + \beta_2$. \\
  Define $\alpha = \min\left\{\frac{1}{2} , \frac{\epsilon}{\beta_2 - \beta_1} \right\}$. \\
  Define $t(\omega) = (1-\alpha) t_1(\omega) + \alpha t_2(\omega)$ . \\
  Let $T = \{t\}$ and $G = \bar{t}$. \tcp*{$\bar{t}$ is the expected payment function, Definition \ref{def:subtangent}}
\end{algorithm}

\begin{remark} \label{remark:contracts-strict-alg}
  One can find a strict subgradient, i.e. a member of $\bar{\partial} c(p_{a^*})$, relatively easily if $(p_{a^*},c_{a^*})$ is a lower vertex.
  The subdifferential $\partial c(p_{a^*})$ is a polytope, and any point in its relative interior will suffice.
  For example, if the polytope is bounded, one can take a uniform convex combination of its vertices.
\end{remark}

Say that a solution is \emph{$\epsilon$-optimal} if it is feasible and has objective value within $\epsilon$ of the optimal solution.
\begin{prop} \label{prop:contracts-compute-strict}
  In a contracts setting, if $a^*$ is strictly elicitable and $\epsilon > 0$, then Algorithm \ref{alg:contracts-strict} computes a menu $G$ that is feasible, that is $\epsilon$-optimal, and that strictly elicits $a^*$.
\end{prop}
\begin{proof}
  \bo{Consider adding the quadratic?}
  Observe that the optimal objective value is $t_1(p_{a^*}) = c(p_{a^*}) + \beta_1$.
  Observe that $\beta_1 \leq \beta_2$ and that the objective value of Algorithm \ref{alg:contracts-strict}'s output is
  \begin{align*}
    t(p_{a^*})
    &= c(p_{a^*}) + (1-\alpha) \beta_1 + \alpha \beta_2  \\
    &=    c(p_{a^*}) + \beta_1 + \alpha (\beta_2 - \beta_1)  \\
    &\leq c(p_{a^*}) + \beta_1 + \epsilon .
  \end{align*}
  (In the case $\beta_2 = \beta_1$, we obtain the same conclusion.)
  It only remains to show that the output is feasible and strictly elicits $a^*$.
  For feasibility, use that the menus $\{t_1\}$ and $\{t_2\}$ are both feasible by construction, and $t$ is a convex combination.
  Strict elicitation follows because $\{t_1\}$ elicits $a^*$, $\{t_2\}$ strictly elicits $a^*$, and $t$ is a convex combination with positive weight on $t_2$.
\end{proof}

\begin{remark} \label{remark:contracts-quadratic}
  In the scoring rule setting, it is natural to prefer \emph{strictly} convex menus $G$ because they lead to \emph{strictly} proper scoring rules.
  If one wishes to avoid ambivalence in the contract setting, and if $a^*$ is strictly elicitable, one can modify the output of Algorithm \ref{alg:contracts-strict} by adding an arbitrarily small quadratic to $G$.
  This results in strict convexity without changing the incentives to elicit $a^*$ strictly.
\end{remark}

\section{Additional figures}


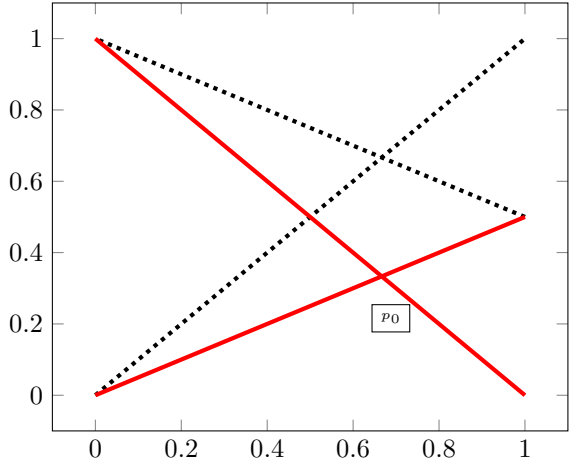
\begin{figure}[H]
\centering
\caption{Comparing our solution for information acquisition to related work. In this example, the signal $S$ completely reveals $\Omega \in \{0,1\}$, which is distributed as a biased coin flip.The black dotted curves are the optimal solutions of related works of \citet{li2020optimization,chen2021optimal}, when the scores are bounded by $1$. The solid red curve is the optimal solution in our model with limited liability when the cost is $\kappa =1/2$.}
\label{fig:compare-related-work}
\begin{subfigure}{.5\textwidth}
  \centering
  \captionsetup{width=.9\linewidth,font=small}
\begin{tikzpicture}
\begin{axis}
    \addplot [mark=none,  black,   ultra thick, dotted] coordinates { (0,0) (1,1)};
    \addplot [mark=none,  black,   ultra thick, dotted] coordinates { (1,0.5) (0,1)};
    \addplot [mark=none,  red,  ultra thick] coordinates { (0,0) (1,0.5)};
    \addplot [mark=none,  red,  ultra thick] coordinates { (1,0) (0,1)};
\end{axis} 
\node [draw] at (4.5,1.5){\tiny $p_0$};
\end{tikzpicture}
\caption{ Our work requires minimizing expected payment with limited liability, i.e. all payments are non-negative. \citet{li2020optimization} places an upper bounded on the payments and maximizes the difference in expected payment of reporting posterior to reporting prior. It so happens in this example that our score is also within the limits but there can be cases where it is not.}
\end{subfigure}%
\begin{subfigure}{.5\textwidth}
  \centering
  \captionsetup{width=.9\linewidth,font=small}
\begin{tikzpicture}
\begin{axis}
    \addplot [mark=none,  black,   ultra thick, dotted] coordinates { (0.65,0) (1,1)};
    \addplot [mark=none,  black,   ultra thick, dotted] coordinates { (0.65,0) (0,1)};
    \addplot [mark=none,  red,  ultra thick] coordinates { (0,0) (1,0.5)};
    \addplot [mark=none,  red,  ultra thick] coordinates { (1,0) (0,1)};
\end{axis} 
\node [draw] at (4.5,0.5){\tiny $p_0$};
\end{tikzpicture}
\caption{ The comparable part of \citet{chen2021optimal} places bounds on the payments and maximizes the agent’s worst-case payoff increment between reporting his posterior prediction and reporting his prior prediction.}
\end{subfigure}%
\end{figure}

\end{document}